\documentclass[aps,pra,amsmath,amssymb,showpacs,preprint]{revtex4-1}
\usepackage{hyperref}
\usepackage[T1]{fontenc}
\usepackage{amssymb,amsmath,braket}
\usepackage{amsfonts}
\usepackage{graphicx}
\usepackage{color}
\usepackage{epstopdf}
\usepackage{subcaption}
\usepackage{bbold}
\usepackage{array}
\usepackage[perpage]{footmisc}

\newcommand{\diag}{\text{diag}}
\newcommand{\tr}[1]{\ensuremath{\text{tr}\left[#1\right]}}
\newcommand{\ptr}[2]{\ensuremath{\text{tr}_{#1}\left[#2\right]}}

\def\ket#1{|#1\rangle}
\def\bra#1{\langle#1|}

\newtheorem{defn}{Definition}[section]
\newtheorem{theo}{Theorem}[section]

\newenvironment{proof}[1][Proof]{\noindent\textbf{#1.} }{\ \rule{0.5em}{0.5em}}

\begin{document}	
\title{Two-qubit causal structures and the geometry of positive qubit-maps} 
\author{Jonas K\"ubler and Daniel Braun}
\affiliation{Institut f\"ur theoretische Physik, Universit\"at T\"{u}bingen,
72076 T\"ubingen, Germany 
}
\begin{abstract}
We study quantum causal inference in a
set-up proposed by Ried et al.~[Nat. Phys. 11, 414 (2015)] in
which a common-cause scenario can be mixed
with a cause-effect scenario, and for which it was found that quantum
mechanics can bring an advantage in distinguishing the
two scenarios: Whereas in classical statistics, interventions such as
randomized trials are needed, a 
 quantum observational scheme 
   can be enough
to detect the causal structure if 
the common cause results from a maximally entangled state.\\
We analyze this setup in terms of the geometry of unital positive but
not completely positive
qubit-maps, arising from the mixture of qubit-channels and steering
maps. We find the range of mixing parameters 
that can generate given correlations,  and prove a quantum advantage
in a more general setup, allowing arbitrary unital channels and initial
states with fully mixed reduced states. This is achieved
by establishing new bounds on signed singular values of sums of
matrices.  
Based on the geometry, we 
quantify and identify the origin of 
the quantum advantage depending on the observed correlations, 
and discuss how additional constraints can 
lead to a unique solution of the problem. 
\end{abstract}

\maketitle

\section{Introduction} 
Imagine a scenario where two experimenters, Alice and Bob, sit in two
distinct laboratories. At one point Alice opens the door of her
laboratory, obtains a coin, checks whether it shows heads or tails and
puts it back out of the laboratory.
Some time later also Bob obtains a 
coin and also he checks whether it shows heads or tails. This
experiment is repeated many times (ideally: infinitely many times) and
after this they 
meet and analyze their joint outcomes. Assuming
their joint probability distribution entails correlations, there must
be some underlying causal mechanism which causally connects their
coins \cite{Reichenbach1971}. This could
be an unobserved
confounder (acting as a common-cause), and 
they actually measured  
two distinct coins 
influenced by the confounder. 
Or it could be that Alice's coin was propagated
by some mechanism to Bob's laboratory, and hence they
actually measured the same coin, with the consequence 
that 
manipulations of the coin by Alice 
can directly influence Bob's result (cause-effect scenario). The task
of  
Alice and Bob is 
to determine the underlying causal structure, i.e.~to distinguish
the two scenarios. This would be rather easy if Alice could prepare
her coin after the 
observation
by her choice and then check whether
this influences the joint probability (so-called ``interventionist
scheme''). In the 
present scenario, however, 
we assume that this is not allowed (so-called ``observational
scheme'').  All that 
Alice and Bob have are therefore the given correlations,  and from
those alone, in general they cannot solve
this task without additional assumptions. 
Ried et al.~\cite{Ried2015} showed that in a similar quantum scenario
involving qubits the above task {\em can} actually be accomplished in 
certain cases even in an observational scheme 
(see below for a
discussion of how the idea of an observational scheme can be
generalized to quantum mechanics). \\

In the present work we consider the
same setup as in \cite{Ried2015}, and allow  
arbitrary convex combinations of the two scenarios: The
common-cause scenario is realized with probability $p$, the
cause-effect scenario with probability $1-p$. Our main result 
are statements about the ranges of the parameter $p$ for which observed correlations 
can be explained with either one of the scenarios, or both. For this,
we cast the problem in the language of affine representations of
\textit{unital positive qubit maps} \cite{Bengtsson2006} in which all
the information is encoded in a  $3\times 3$ real matrix, as is standard
in quantum information theory for \textit{completely positive
  unital qubit maps} \cite{Nielsen2009}. \\  

The paper is structured as follows: In section \ref{sec:Causal} we
introduce causal models for classical random variables and for
quantum systems. Therein we define what 
we consider a
\textit{quantum observational scheme}. Section \ref{sec:positiveMaps}
introduces the mathematical framework of ellipsoidal
representations of qubit quantum-channels and qubit steering-maps. In
section \ref{sec:Results} we define 
our problem mathematically and prove the
main results, which we then comment in the last section
\ref{sec:discussion}.

\section{Causal inference: classical versus quantum} \label{sec:Causal}
\subsection{Classical causal inference}
At the heart of a \textit{classical causal model} is a set of random
variables $X_1, X_2, ..., X_N$. The observation of a specific value of
a variable, $X_i = x_i$, is associated with an \textit{event}.
Correlations between events hint at 
some kind of causal mechanism that links the events \cite{Reichenbach1971} 
. Such a mechanism
can be a deterministic law as for example $x_i = f(x_j)$ or can be a
probabilistic process described by conditional probabilities
$P(x_i|x_j)$, i.e.~the probability to find $X_i=x_i$  given 
$X_j=x_j$ 
was observed. The causal mechanism may not be 
merely a direct causal influence from one observed event on the other,
but may be
due to common causes that lead  with a certain
probability to both events --- or a mixture between both
scenarios. Hence, by merely analysing correlations 
$P(x_1,x_2,\ldots,x_n)$, i.e.~the joint probability distribution of
all events, 
one can, in general, without prior 
knowledge of the \textit{data generating 
  process},   not uniquely determine  the causal mechanism that leads
to the observed correlations (purely \textit{observational} scheme).  
To remedy this, an intervention is often necessary, where the value of a
variable $X_i$ whose causal influence one wants to investigate, is set
by an experimentalist to different values, trying to see whether this
changes the statistics of the remaining events (\textit{interventionist}
scheme). 
One strategy
for reducing the influence of other, unknown factors, is to randomize the samples.
This is  for example a typical approach in clinical studies, where one
group of randomly selected probands receives a treatment whose
efficiency one wants to investigate, and a randomly selected control
group receives a placebo. If the percentage of cured people in the
first group is significantly larger than in the second group, one can
believe in a positive causal effect of the treatment.
The probabilities obtained in this interventionist scheme are
so-called ``do-probabilities''  (or ``causal  
conditional probabilities'') \cite{Pearl2009a}:
$P\left(x_i|\text{do}(x_j)\right)$ is the probability to find
$X_i=x_i$ if an 
experimentalist intervened and set the value of $X_j$ to the value
$x_j$.  This is different from $P\left(x_i|x_j\right)$, as a
possible causal influence from some other unknown event on $X_j=x_j$
is cut, i.e.~one deliberately modifies the underlying causal structure
for better understanding a part of it.  If $X_j=x_j$ was the only
direct cause of $X_i=x_i$ then 
$P\left(x_i|x_j\right) = P\left(x_i|\text{do}(x_j)\right)$. If instead 
the event $X_i=x_i$ was a cause of $X_j=x_j$, then intervening on $X_j$
cannot change $X_i$: $P(x_i)=
P\left(x_i|\text{do}(x_j)\right)=P\left(x_i|\text{do}(\bar{x_j})\right)$,
where $\bar{x_j}$ is a value different from $x_j$.   If the
correlation between $X_i=x_i$ and 
$X_j=x_j$ is purely because of a common cause, then no intervenion on
$X_i$ or $X_j$ will change 
the probability to find a given value of the other:
$P(x_i) = P\left(x_i|\text{do}(x_j)\right)$ for all $x_j$, and $P(x_j)
= P\left(x_j|\text{do}(x_i)\right)$ for all $x_i$.  Observing these
do-probabilities one can hence draw conclusions about the causal
influences behind  the correlations observed in the occurence of
$X_i=x_i$ and $X_j=x_j$. \\

 In practice, direct 
causation in one direction is often excluded by time-ordering and need not
to be investigated. For example, when doubting that one can conclude
that smoking causes
lung cancer from the observed correlations between these two
events, it does not make sense to claim that having lung cancer causes
smoking, as usually smoking comes before developing lung
cancer. 
But even dividing a large number of people randomly into two
groups and forcing one of them to smoke and the other not to smoke in
order to find out if there is a common cause for both
would be ethically 
inacceptable.  The needed do-probabilities can
therefore not 
always be obtained by experiment. Interestingly, the causal-probability calculus
allows one in certain cases, depending notably on the graph structure,
to 
{\em calculate} do-probabilities from 
observed correlations without having to do the intervention. 
Inversely, apart from only
predicting the conditional probabilities for a random
variable say $X_i$ given the 
observation of $X_j=x_j$, denoted as $P(x_i|x_j)$, a causal model can
also predict the do-probabilities, i.e.~the distribution of $X_i$ if
one \textit{would} {intervene} 
on the variable $X_j$ and set its value to $x_j$. 
This is crucial for deriving informed recommendations for actions
targeted at modifying certain probabilities, e.g.~recommending not to
smoke in order to reduce the risk for cancer. \\ 

The structure of a causal model can be depicted by a
graph. Each random variable is represented by a vertex of the graph.
Causal connections are represented by directed arrows and imply that
signaling along the direction of the arrow is possible. In a classical
causal model it is assumed that events happen at specific points in
space and time, therefore bidirectional signaling is not possible as
it would imply signaling backward in time. Hence the graph cannot
contain cycles and is therefore a \textit{directed acyclic graph}
(DAG) \cite{Pearl2009a}, see FIG.\ref{Fig:4-DAG}. The set of
parents $PA_j$ of the random variable $X_j$ is defined as the set of
all variables that have an immediate arrow pointing towards $X_j$, and
$pa_j$ denotes a possible value of $PA_j$. The causal model is then
defined through its graph with random variables $X_i$ at its vertices 
and the weights $P(x_j|pa_j)$ of each edge, i.e.~the probabilities
that $X_j=x_j$ happens under the condition 
that $Pa_j=pa_j$ occurred. 
The model generates the entire correlation function
according to
\begin{equation}
  \label{eq:Pcm}
  P(x_1,\ldots,x_n)=\prod_{j=1}^nP(x_j|pa_j)\,,
\end{equation}
which is referred to as \textit{causal Markov condition} \cite{Pearl2009a}.
When
all $P(x_1,\ldots,x_n)$ are  
given, then all conditional probabilities follow, hence all
$P(x_j|pa_j)$ that appear in a given graph, but in general not all
correlations nor all $P(x_j|pa_j)$ are known (see below). 
The causal inference probleme consists in finding a graph structure
that allows one to satisfy eq.\eqref{eq:Pcm} for given data
$P(x_1,\ldots,x_n)$ and all 
known $P(x_j|pa_j)$, where the unknown $P(x_j|pa_j)$
can be considered fit-parameters in case of incomplete data.  
With access to the full joint
probability distribution, the causal inference only needs to
determine the graph.  
In practice, 
however, one often has only incomplete data: as long as a common cause
has not been determined yet, 
one will not have data involving correlations of the corresponding
variable.  For example, one
may have strong correlations between getting lung cancer
(random variable $X_2\in\{0,1\}$) and smoking (random variable
$X_1\in\{0,1\}$),  
but if there is a unknown common cause $X_0$ for both, one typically
has 
no information about $P(x_0,x_1,x_2)$: One will only start collecting
data about correlations 
between the presence of a certain gene, say,  and the habit of smoking or
developing lung cancer once one suspects that gene to be a 
cause for at least one of these. In this case $P(x_1|x_0)$ and
$P(x_2|x_0)$ are fit parameters to the model as well. The possibility
of extending a causal model through inclusion of unknown random
variables is one reason why in general there is no unique solution to
the causal inference problem based on correlations
alone. Interventions on $X_i$ make it possible, on the other hand, to
cut $X_i$ from its parents and hence eliminate unknown causes one by
one for all random variables. \\

Once a causal model is known, one can
calculate all distributions  
\begin{align}
P(x_1,...,x_n \vert i_1,..., i_n) &= \prod_{j=1}^n P(x_j \vert pa_j, i_j), \label{eq:causalMarkov}
\end{align}
for all possible combinations of interventions and observations, where
the $i_j$  are the values of the intervention variable $I_j$ for 
the event $X_j$, $i_j = \text{idle}$ or $i_j = \text{do}(x_j)$. Here,
$P\left(x_j|  pa_j,i_j =\text{do}(\tilde{x}_j) \right)=
\delta_{x_j,\tilde{x}_j}$ reflects that an intervention on $X_j$
deterministically sets its value, independently of the 
observed values of its causal parents.
If $I_j = \text{idle}$ then the value of $X_j$ only depends on its
causal parents $PA_j$, i.e. $P(x_j|  \{x_i\}_{i\neq j},i_j =
\text{idle}) = P(x_j|  pa_j,i_j = \text{idle})$. 
\begin{figure}
\centering
\includegraphics[scale=.2]{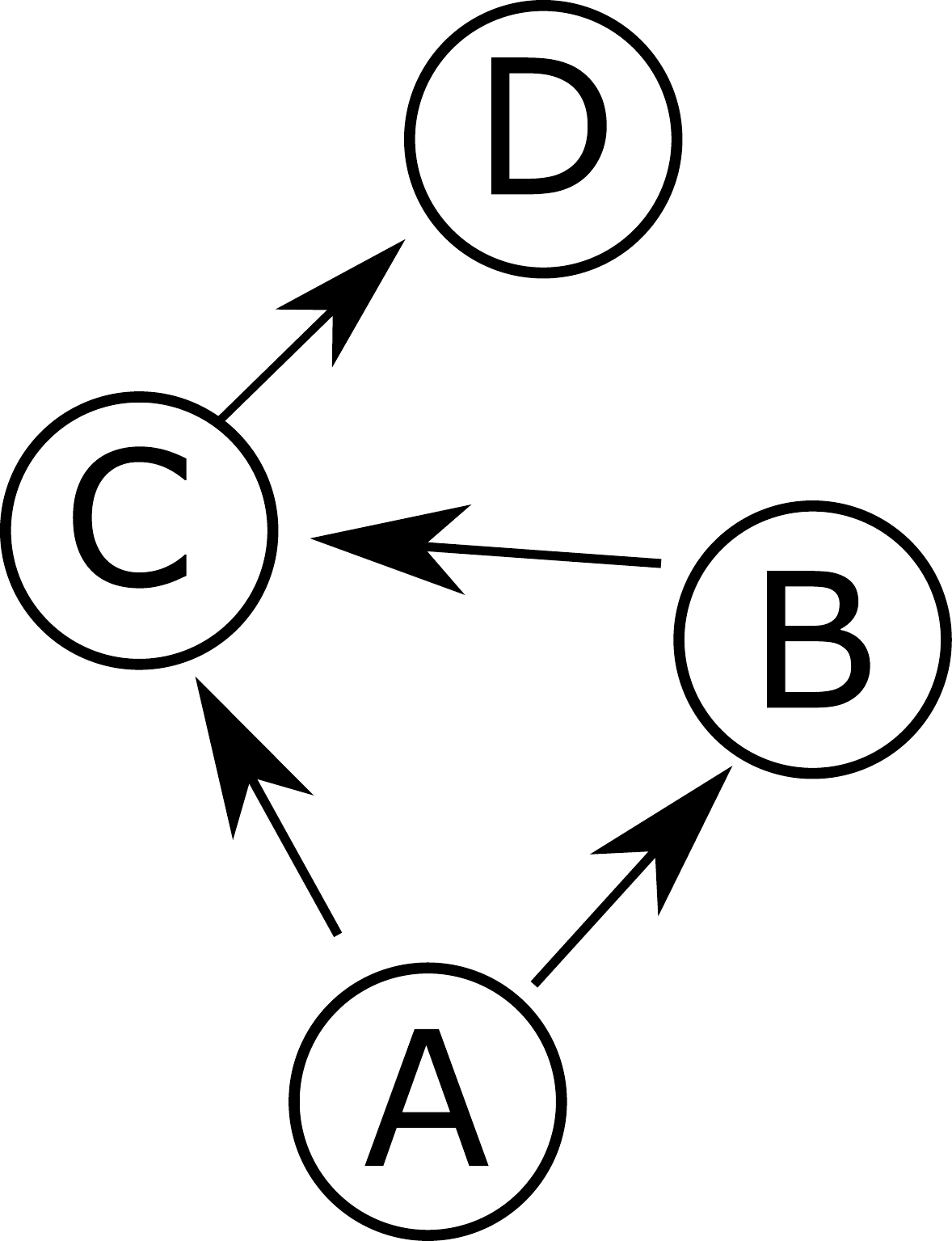}
\caption{Simple DAG in a four party scenario. The parental structure is $\text{PA}_A = \{\},\; \text{PA}_B = \{A\},\; \text{PA}_C = \{A,B\},\;\text{PA}_D = \{C\}$. According to the causal Markov condition, eq.~\eqref{eq:causalMarkov}, the probability distribution then factorizes as ${P(a,b,c,d|i_A, i_B, i_C, i_D) = P(d|c, i_D) P(c|a,b, i_C) P(b|a, i_B) P(a|i_A)}$.}
\label{Fig:4-DAG}
\end{figure}

The field of \textit{causal discovery} or \textit{causal inference}
aims at providing methods to determine the causal model, that is the
                                DAG and the joint-probability
                                distributions entering  
                                \eqref{eq:Pcm} 
 for a given scenario. Different combinations of the $I_j$ correspond
 to different strategies.  If all the interventions are set to
 idle, and hence all the outcomes are determined by the causal
 parents, one has the purely {observational} approach. 
In multivariate scenarios, where more than two random variables are involved,
the observation 
 of the joint probability distribution alone can still contain hints of the
 causal structure 
 based on conditional independencies \cite{Pearl2009a}.  
Nevertheless, in the bivariate scenario, i.e.~when only two random
variables are involved, classical  correlations obtained by
observations do not comprise any causal information. Only if
assumptions for example on the noise distribution are taken a priori,
information on the causal model can be obtained from observational
data \cite{Mooij2016}.  
 
\subsection{Quantum causal inference}
The notion of causal models does not easily translate to quantum
mechanics. The main problem is that in quantum systems not all
observables can have
predefined values independent of observation. 
Similiar to an operational formulation of quantum mechanics
\cite{Chiribella2011}, 
 the \textit{process matrix formalism} was introduced \cite{Oreshkov}
and a quantum version of an \textit{event} defined. 
In
\cite{Costa2016} this is reviewed for the purpose of causal models. In
place of the random variables in the classical case there are local
\textit{laboratories}. 
Within a process each laboratory obtains a quantum system as input and
produces a quantum system as output. 
A quantum event corresponds to information
which is obtained within a laboratory and is associated with a
\textit{completely positive} (CP) map mapping the input Hilbert space
to the output Hilbert space of the laboratory. The possible events
depend on the choice of \textit{instrument}. An instrument is a set of
CP maps that sum to a \textit{completely positive trace preserving}
(CPTP) map. For example an instrument can be a projective measurement
in a specific basis, with the events the possible outcomes. The
possibility to choose different instruments mirrors the possibility of
interventions in the classical case \cite[3.3]{Costa2016}. The whole
information about mechanisms, which are represented as CPTP maps, and
the causal connections is contained in a so-called \textit{process
  matrix}. Besides its analogy for a classical causal model, the process
framework goes beyond classical causal structures as it does not
assume such a fixed causal structure \cite{Oreshkov}. This recently
stirred a lot of research \cite{Oreshkov2016, Procopio2015,
  Chiribella2012, Guerin}.
For a more detailed introduction we refer the reader
especially to reference \cite{Costa2016} where a comprehensive
description is provided.\\  

The analogue of causal inference in the classical case is the
reconstruction of a process matrix. This can be done using
informationally complete sets of instruments, theoretically described
in \cite[4.1]{Costa2016} and experimentally implemented in
\cite{Ried2015}. 
Defining a \textit{quantum observational scheme} in
analogy to the classical one 
is not straight forward. In 
general a quantum measurement destroys much of the states' character
and hence can almost never be considered a passive observation. For
example if the system was initially in a pure state $\ket{\psi}$ 
but one measures in a basis such that $\ket{\psi}$ is not an eigenstate
of the projectors onto the basis states, then the measurement truly
changes the state of the 
system and the original state is not reproduced in the
statistical  average. 
In
\cite[sect. 5]{Costa2016} an observational scheme is simply defined as
projective measurements in a fixed basis,  
in particular without
assumptions 
about the 
incoming state of a laboratory and thus without assumptions about the underlying process.
 Another possibility to define an observational scheme 
is based on the idea  that in the classical world 
observations reveal pre-existing properties of physical systems and
that quantum observations should 
reproduce this. 
As a consequence, if
one mixes
the post-measurement states with the probabilities of the
corresponding measurement outcomes, one
should obtain the same state as before the measurement.  That is
ensured if and only if 
operations that do not destroy the quantum character of the
state are allowed, as  
coherences cannot be restored by averaging. 
Ried et al.~\cite{Ried2015} formalized this notion as ``informational
symmetry'', but considered only  preservation of local states.
 For the special case of locally
completely mixed states, 
they showed that projective measurements in arbitrary bases
possess informational symmetry.
This definition of a quantum observational scheme 
is problematic due to two reasons:
Firstly, the allowed class of instruments 
depends on the 
incoming state, i.e.~one can only apply
projective measurements 
that are diagonal in the same basis as the
state itself. This is at variance with the typical motivation for an
observational scheme,
namely that the instruments are restricted
a-priori due to practical reasons. Moreoever, having measurements
depend on the state requires prior knowledge about the state of the
system, but finding out the state of the system is part of the causal
inference (e.g.: are the correlations based on a state shared by
Alice and Bob?). Hence, in general one cannot assume sufficient
knowledge of the state for restricting the measurements such that they
do not destroy coherences.  
\\  
Secondly, the definition is 
unnaturally restrictive 
as it only considers the local state and not the global state. 
For example if Alice and Bob share a singlet state
$\ket{\psi} = \frac{\ket{01}-\ket{10}}{\sqrt{2}}$, then both local
states are completely mixed.  Hence according to the informational
symmetry,  they are allowed to perform
projective measurements in arbitrary bases. 
If Alice and Bob now both measure in the computational basis, they
will each obtain both outcomes with probability $1/2$ and their local
states will remain invariant in the statistical average $\rho_A' =
\rho_A = \frac{\mathbb{1}}{2}=\rho_B' =\rho_B$. However, the global
state does not remain intact. The post-measurement state is given as
$\rho_{AB}' = \frac{1}{2} \left( \ket{01}\bra{01} +
  \ket{10}\bra{10}\right)$ which is not even entangled anymore. But
even defining a  ``global informational symmetry'', i.e.~requiring the
global state to remain invariant,  does not settle the issue in a
convenient way, as this would not allow any local measurements of
Alice and Bob.\\ 

Here we propose three different schemes ranging from  full quantum
interventions over a quantum-observational scheme with the
possibility of an active
choice of measurements, to a passive quantum observational scheme in a
fixed basis that comes closest to the classical observational scheme.

\begin{table}
\begin{tabular}{{|c||>{\centering\arraybackslash}p{2.5cm}|>{\centering\arraybackslash}p{2.5cm}|>{\centering\arraybackslash}p{2.5cm}||>{\centering\arraybackslash}p{2.5cm}|>{\centering\arraybackslash}p{2.5cm}|}}
\hline
& arbitrary \mbox{instruments}  & arbitrary projections & fixed basis projection &signaling& causal inference\\ \hline
Q-interventionist & $\surd$  & $\surd$ &$\surd$ &$\surd$ & $\surd$\\
Active Q-observational & X & $\surd$ & $\surd$  &$\surd^1$ & $(\surd)^2$\\
Passive Q-observational & X&X& $\surd$&X&X$^3$\\
\hline
\end{tabular}
\caption{\textbf{Quantum schemes for causal inference}: An overview of
  instruments allowed within different quantum schemes defined in this
  section. $\surd$ indicates allowed/possible, X indicates not
  allowed/impossible. $^1$In the active quantum observational scheme
  signaling is possible in principle.  However, in the scenarios
  considered in this work signaling is not possible, and still causal
  inference can be successful. $^2$The potential of causal inference
  in the active quantum-observational scheme is discussed in the main part of
  this paper. $^3$In the passive quantum-observational scheme no more causal
  inference than classical is possible. 
}  
\end{table}

The definitions are based on restricting the allowed set
of instruments. An instrument is to be understood in the
process-matrix context. In all three schemes the set of allowed
instruments is 
independent of the actual underlying processes, which is a reasonable
assumption, since the motivation for 
 causal inference comes from the fact that states or processes are not
 known in the first place.
\begin{itemize}
\item[]\textbf{Quantum interventionist scheme:} Arbitrary instruments
  can be applied in local laboratories. These include for example deterministic
  operations such as state preparations 
 or simply projective measurements.
  An appropriate choice of the instruments enables one
to detect causal structure in
  arbitrary scenarios, i.e.~to reconstruct the process
  matrix \cite{Costa2016}. This scheme resembles most closely an
  interventionist scheme in a classical scenario
but offers additional quantum-mechanical possibilities of intervention.
  
\item[]\textbf{Active quantum-observational scheme:} Only projective
  measurements 
  in arbitrary orthogonal bases 
are allowed, but no post-processing of the
  state after the measurement. 
The latter request translates the
idea of not intervening in 
  the quantum realm, as it is not possible to 
  deterministically  
  change the state by the experimenters choice.
  Depending on the state and the instrument, the
  state may change during the measurement, hence the scheme is invasive, but the
  difference to the classical observational scheme arises solely from
  the possible destruction of quantum coherences. 
    This is a quantum
  effect without classical correspondence and hence opens up a new
  possibility of defining an observational scheme that has no
  classical analogue.
    Repetitive
  application of the same measurement within a single run
always gives the same output. Furthermore, we allow projective
measurements in different bases in different runs
of the experiment. This 
freedom allows one to completely characterize the incoming
state. \\  
This scheme allows for signaling, i.e.~there exist
processes for which Alice's choice of instrument changes the
statistics that Bob observes. As an example consider the process,
where Alice always obtains a qubit in the state 
 $\ket{1}$.  She applies her
instrument on it, and then the outcome is propagated to Bob by the
identity channel. Bob measures in the basis where $\ket{1}$ is an
eigenstate. If Alice measured in the same basis as Bob, then both of
them deterministically obtain 1 
as result. If Alice  instead
measures in the basis 
$\left\{\ket{\pm} = \frac{1}{\sqrt{2}} 
(\ket{0}\pm\ket{1})\right\}$, then Bob would obtain 
1 only with
probability $\frac{1}{2}$. This is considered as 
signaling according 
to the definition in \cite{Costa2016}. Clearly, signaling presents a
direct quantum advantage for causal inference compared to a classical
observational 
scheme, and motivates the attribute ``active'' of the scheme. 
In the present 
work we focus on this scheme, but 
exclude such a direct quantum advantage
by considering exclusively unital channels and a completely 
mixed incoming state for Alice,  as was done also in
\cite{Ried2015}.
It is then impossible for Alice to send a signal to Bob if her
instruments are restricted to 
quantum observations, even if she is allowed to actively 
set
her measurement basis. 
One might wonder whether the quantum-observational scheme can be
generalized to POVM measurements. However, these do not fit into the
framework of 
instruments that transmit an input state to an output state, as POVM measurements
do not specify the post-measurement state.

\item[]\textbf{Passive quantum-observational scheme:}
For the whole setup a fixed 
basis is selected. Only projective
measurements with respect to this basis are 
permitted, and it is 
forbidden to change the basis 
in different runs of the experiment. This is also what is used in
\cite{Costa2016} to obtain 
classical causal models as a limit of quantum causal models. 
Since the
basis is fixed independently of the 
underlying process, the
measurement can still be invasive in the sense that it can destroy
coherences, and  hence it is still not a pure observational scheme in the
classical sense.  
Nevertheless, Alice cannot signal to Bob here  as she
has no possibility of actively encoding
information in the quantum state, regardless of the nature of the
state, which motivates the name ``passive
quantum-observational scheme''. As without any change of basis it is
impossible to exploit 
stronger-than-classical quantum correlations, this
scheme comes closest to a classical observational scheme. 
And due to the restriction
to observing at most classical 
correlations, it is not possible to infer anything more about the
causal structure than classically possible. 
\end{itemize}

\section{Affine representation of quantum channels and steering
  maps} \label{sec:positiveMaps} 
In this section we introduce the 
tools of quantum
information theory 
that we need to analyze the problem of causal inference in section
\ref{sec:Results}. 

\subsection{Bloch-sphere representation of qubits}
A qubit is a
quantum system with a two-dimensional Hilbert space with basis states  denoted as $\ket{0}$ and $\ket{1}$. An 
arbitrary state of 
the qubit is described by a density operator $\rho$, 
a positive linear operator with unit trace, $\rho \geq 0, \; \tr{\rho} =
1$. Every single-qubit state can be represented geometrically by its
\textit{Bloch-vector} $\boldsymbol{r} = \tr{\rho
  \boldsymbol{\sigma}}$, with $|\boldsymbol{r}| \leq 1$   as 
\begin{align}
\rho = \frac{\mathbb{1}+\boldsymbol{r}\cdot \boldsymbol{\sigma}}{2},
\end{align}
where $\boldsymbol{\sigma} = (\sigma_1, \sigma_2, \sigma_3)^T$ denotes the vector of Pauli matrices.  

\subsection{Channels}
A quantum channel $\mathcal{E}$ is 
a \textit{completely positive trace preserving map} (CPTP map). A 
quantum channel
maps a density operator in the space of linear operators $\rho \in \mathcal{L}(\mathcal{H})$ on the Hilbert space $\mathcal{H}$ to a density operator in the space of linear operators $\rho' \in \mathcal{L}(\mathcal{H}')$ on a (potentially different) Hilbert space $\mathcal{H}'$.
\begin{align*}
\mathcal{E} : \rho \rightarrow \mathcal{E}(\rho) \equiv \rho', \qquad \rho, \rho' \geq 0, \; \tr{\rho}= \tr{\rho'} = 1.
\end{align*} 
This formalism describes any physical dynamics of a quantum
system. Every quantum channel can be understood as the unitary
evolution of the system coupled to an environment \cite{Nielsen2009}.  
The 
constraint of complete positivity can be understood the following
way. If we extend the map $\mathcal{E}$ with the identity operation of
arbitrary dimension, the composed map $\mathcal{E} \otimes
\mathbb{1}$, which acts on a larger system, should still be
positive. 
An example of a map that is positive but not completely positive is
the transposition 
map, that, if extended to a larger system, maps entangled states to
non-positive-semi-definite operators 
\cite[chapter 11.1]{Bengtsson2006}. \\ 

\textit{Geometrical representation of qubit maps}\\
Every qubit channel (a quantum channel mapping a qubit state onto a
qubit state)  $\mathcal{E}$ can  be described completely by its action on the Bloch sphere, see \cite{Fujiwara1999,Braun2014, BethRuskai2002} and is completely described by the matrix $\Theta_\mathcal{E}$ mapping the 4D Bloch vector $(1,\boldsymbol{r})$,
\begin{align}
\Theta_\mathcal{E} = \begin{pmatrix}
1 & 0 \\
\boldsymbol{t_\mathcal{E}} & T_\mathcal{E}
\end{pmatrix},
\end{align}
where the upper left 1 ensures trace preservation.
A
state $\rho$ described by its Bloch vector $\boldsymbol{r}$ is then
mapped by the quantum channel $\mathcal{E}$ to the new state $\rho'$
with Bloch vector  
\begin{align*}
\boldsymbol{r}' = T_\mathcal{E} \boldsymbol{r} +\boldsymbol{t_\mathcal{E}}.
\end{align*}
A qubit channel is called \textit{unital} if it leaves the completely
mixed state  invariant: $\mathcal{E}(\rho_\text{mixed}) = \rho_\text{mixed}$, with $\rho_\text{mixed} = \frac{\mathbb{1}}{2}$,
i.e.~$\boldsymbol{r}_\text{mixed} = \boldsymbol{0}$. For unital
channels $\boldsymbol{t}_\mathcal{E}$ vanishes. The whole information
is then contained in the 3x3 real matrix $T_\mathcal{E}$, 
which we refer to as \textit{correlation matrix} of the channel. The  matrix
$T$ (from now on we
drop the index $\mathcal{E}$) can be expressed by writing it in its
signed singular value decomposition \cite[eq. (9)]{Braun2014},
\cite[eq. (10.78)]{Bengtsson2006} (see also the appendix around
equation  \eqref{def:App_SSV}),
\begin{align}
T =R_1 \eta R_2. \label{eq:signed_singulars}
\end{align}
Here, $R_1$ and $R_2$ are  proper rotations (elements of the $SO(3)$
group), corresponding to unitary channels, that is $R_iR_i^T =
\mathbb{1}$ with $\det(R_i)=1$, and $\eta= \diag(\eta_1,\eta_2,
\eta_3)$ is a real diagonal matrix. This can be interpreted rather
easily. A unital qubit channel maps the Bloch sphere onto an
ellipsoid, centered around the origin, that fits inside the Bloch
sphere. First the Bloch sphere is rotated by $R_2$ than it is
compressed along the coordinate axis by factors $\eta_i$. The
resulting ellipsoid is then again rotated. Hence, apart from unitary
freedom in the input and output, the unital 
quantum channel is completely
characterized by its 
\textit{signed singular values} (SSV) $\bm \eta$
\cite[II.B]{Braun2014}. The CPTP property gives restrictions to the
allowed values of 
 $\boldsymbol{\eta} \equiv (\eta_1,\eta_2, \eta_3)^T$. These are commonly known as the \textit{Fujiwara-Algoet} conditions 
 \cite{Fujiwara1999,Braun2014} 
 \begin{align}
 \begin{aligned}
 1+\eta_3 \geq |\eta_1+\eta_2|,\\
1-\eta_3 \geq |\eta_1-\eta_2|.
 \end{aligned}\label{eq:Fujiwara}
\end{align}  The allowed values for $\boldsymbol{\eta}$ lie inside a tetrahedron $\mathcal{T}_\text{CP}$ (the index CP stands for completely positive),
\begin{align}
\mathcal{T}_\text{CP} \equiv \text{Conv}\left(\left\lbrace\boldsymbol{v}^\text{CP}_i \right\rbrace_i\right), \label{def:T_CP}
\end{align} 
where $\text{Conv}\left(\left\lbrace x_i\right\rbrace_i\right) \equiv \left\lbrace \sum_i p_i x_i | p_i\geq 0, \sum_i p_i =1\right\rbrace$ denotes the convex hull of the set $\left\{x_i\right\}_i$ and the vertices are defined as, 
\begin{align}
\begin{aligned}\label{eq:CP_vertices}
\boldsymbol{v}^\text{CP}_1 &= (1,1,1)^T,\\
\boldsymbol{v}^\text{CP}_2 &= (-1,-1,1)^T,\\
\boldsymbol{v}^\text{CP}_3 &= (-1,1,-1)^T,\\
\boldsymbol{v}^\text{CP}_4 &= (1,-1,-1)^T.
\end{aligned}
\end{align}

For a more detailed discussion of qubit maps we refer the reader to chapter 10.7 of \cite{Bengtsson2006}.

%
 \subsection{Steering}
In quantum mechanics, measurement outcomes on two spatially separated
partitions of a composed quantum system can be highly correlated \cite{Bell1964}, and further the choice of measurement
operator on one side can strongly influence or even determine the
outcome on the other side \cite{Schrodinger1935}, a phenomenon known
as ``steering''. Suppose Alice and
Bob share the two qubit state $\rho_{AB}$. If Alice performs a
measurement on it, leaving her 
qubit in the state $\rho_A$ then Bob's qubit is steered to the state $\rho_B$
proportional to the (unnormalized) state $ \ptr{A}{\rho_{AB} (\rho_A
  \otimes \mathbb{1})}$ \cite[p.2]{Jevtic}.  This defines a positive
linear trace preserving map $\mathcal{S}: \rho_A \rightarrow
\mathcal{S}(\rho_A) = \rho_B$, called \textit{steering map}, that
depends on the state $\rho_{AB}$.
 \\
Steering maps have been intensely studied especially in terms of entanglement characterization \cite{Jevtic,Milne2014}. In analogy to the treatment of qubit channels, we can associate an unique ellipsoid inside the Bloch sphere with a two-qubit state, known as steering ellipsoid, that encodes all the information about the bipartite state \cite{Jevtic}.\\
 Every bipartite two qubit state can be expanded in the Pauli basis as
\begin{align*}
\rho_{AB} = \frac{1}{4} \sum_{\mu, \nu =0}^3 \Theta_{\mu \nu} \sigma_\nu \otimes \sigma_\mu,
\end{align*}
where 
\begin{align}
\Theta_{\mu\nu} = \tr{\rho_{AB} \sigma_\nu \otimes \sigma_\mu}.
\end{align}
Note that we defined $\Theta$ to be the transposed of the one defined in \cite{Jevtic}, since we want to treat steering from Alice to Bob. The matrix contains all the information about the bipartite state and can be written as
\begin{align*}
\Theta = \begin{pmatrix}
1 & \boldsymbol{a}^T \\
\boldsymbol{b} & T_\mathcal{S}
\end{pmatrix},
\end{align*}
where $\boldsymbol{a}$ ($\boldsymbol{b}$) denotes the Bloch vector of Alice's (Bob's) reduced state. $T_\mathcal{S}$ is a 3x3 real orthogonal matrix and encodes all the information about the correlations, 
and we will refer to it as \textit{correlation matrix} of the steering map.\\
In this work we 
only consider bipartite qubit states which have completely mixed
reduced states $\ptr{A}{\rho_{AB}} = \ptr{B}{\rho_{AB}} =
\mathbb{1}/2$ or equivalently  $\boldsymbol{a}=
\boldsymbol{b}=\boldsymbol{0}$. In analogy to unital channels we 
call such states \textit{unital two-qubit states} and the
corresponding maps \textit{unital steering maps}. Up to local unitary
operations on the two partitions, the 
correlation matrix $T_\mathcal{S}$ is
characterized by its signed singular values
$\eta_1,\eta_2,\eta_3$. The allowed values of these are given through
the positivity constraint on the density operator 
$\rho_{AB}$ defined up to local unitaries as (cf.~equation (6) in
\cite{Milne2014}) 
\begin{align}
\rho_{AB}  = \frac{1}{4}\left(\mathbb{1}\otimes \mathbb{1} + \sum_{i=1}^3 \eta_i \sigma_i\otimes \sigma_i\right). 
\end{align}
The positivity of $\rho_{AB}$ implies the conditions (the derivation is analogue to the derivation of (10)-(15) in \cite{Braun2014})
\begin{align}
\begin{aligned}
1+\eta_3 \geq |\eta_1-\eta_2|,\\
1-\eta_3 \geq |\eta_1+\eta_2|.
\end{aligned}
\end{align}
These are the same as for unital qubit channels
(eq.~\eqref{eq:Fujiwara}) up to a sign flip, 
and define the tetrahedron $\mathcal{T}_\text{CcP}$ of unital \textit{completely co-positive trace preserving maps} (CcPTP) \cite{Bengtsson2006,Braun2014},
\begin{align}
\mathcal{T}_\text{CcP} \equiv \text{Conv}\left(\left\lbrace\boldsymbol{v}^\text{CcP}_i \right\rbrace_i\right), \label{def:T_CcP}
\end{align}
with the vertices
\begin{align}\label{eq:CcP_vertices}
\begin{aligned}
\boldsymbol{v}^\text{CcP}_1 &= (-1,-1,-1)^T,\\
\boldsymbol{v}^\text{CcP}_2 &= (-1,1,1)^T,\\
\boldsymbol{v}^\text{CcP}_3 &= (1,-1,1)^T,\\
\boldsymbol{v}^\text{CcP}_4 &= (1,1,-1)^T.
\end{aligned}
\end{align}

CcPTP maps are exactly CPTP maps with a preceding transposition map,
i.e.~for every steering map $\mathcal{S}$ there exists a quantum
channel $\mathcal{E}$ such that $\mathcal{S} = \mathcal{E}\circ
\mathcal{T}$, where $\mathcal{T}$ is the transposition map  
with respect to an arbitrary but fixed basis (see e.g.~\cite{Bengtsson2006}).
 \begin{figure}
\begin{subfigure}{0.3\textwidth}
\centering
\includegraphics[scale=0.33]{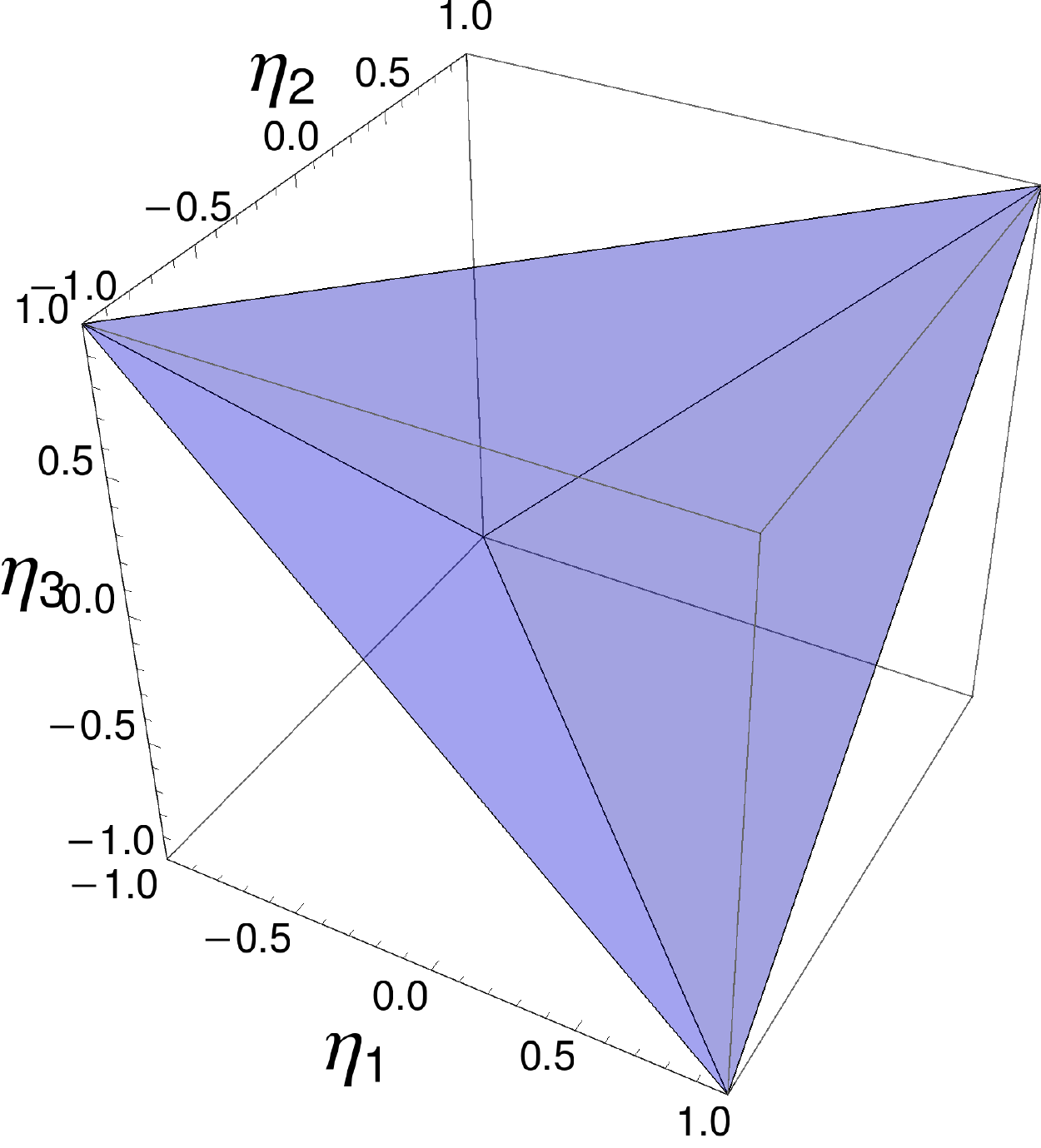}
\end{subfigure}
\begin{subfigure}{0.3\textwidth}
\centering
\includegraphics[scale=0.33]{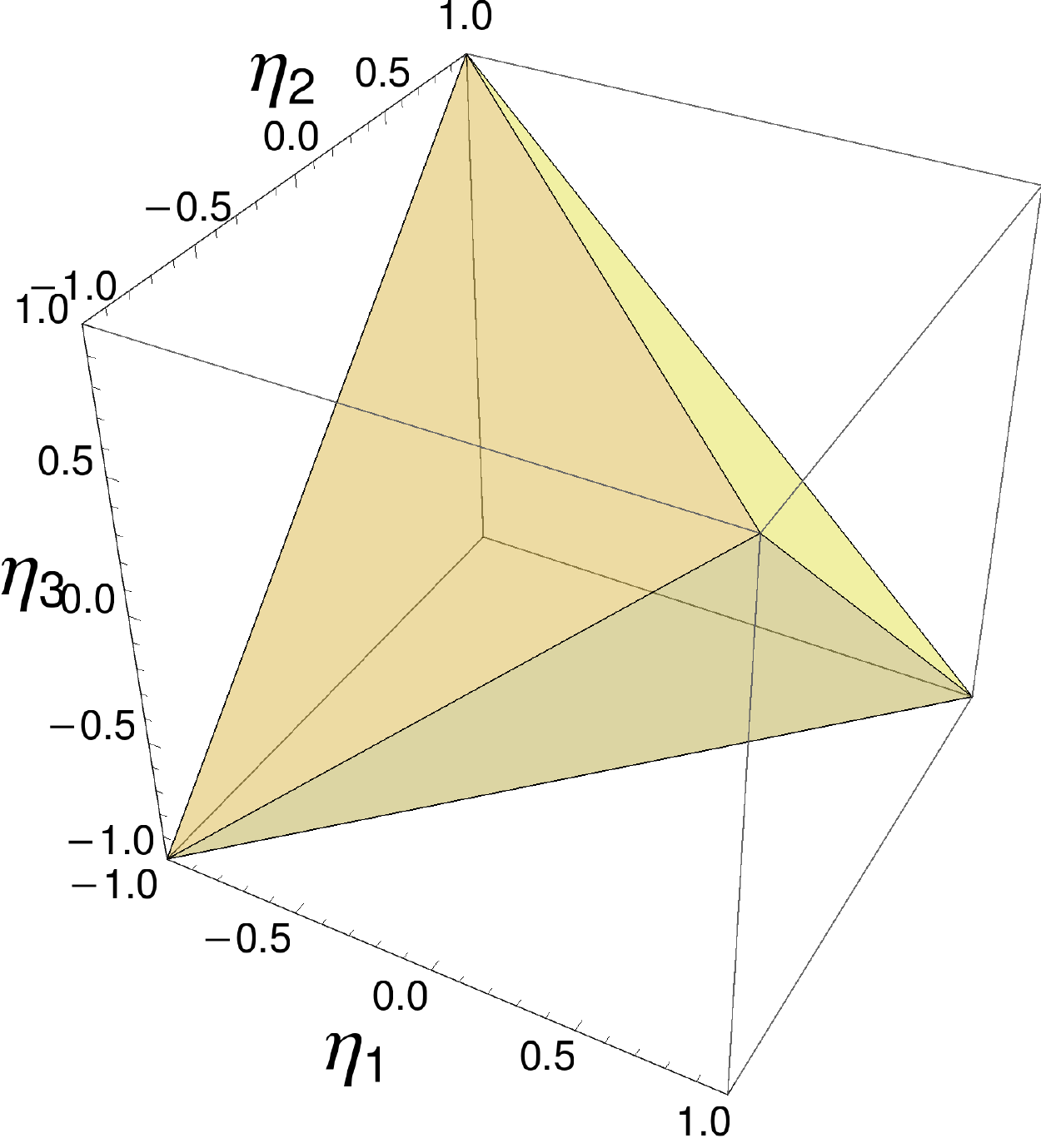}
\end{subfigure}
\begin{subfigure}{0.3\textwidth}
\centering
\includegraphics[scale=0.33]{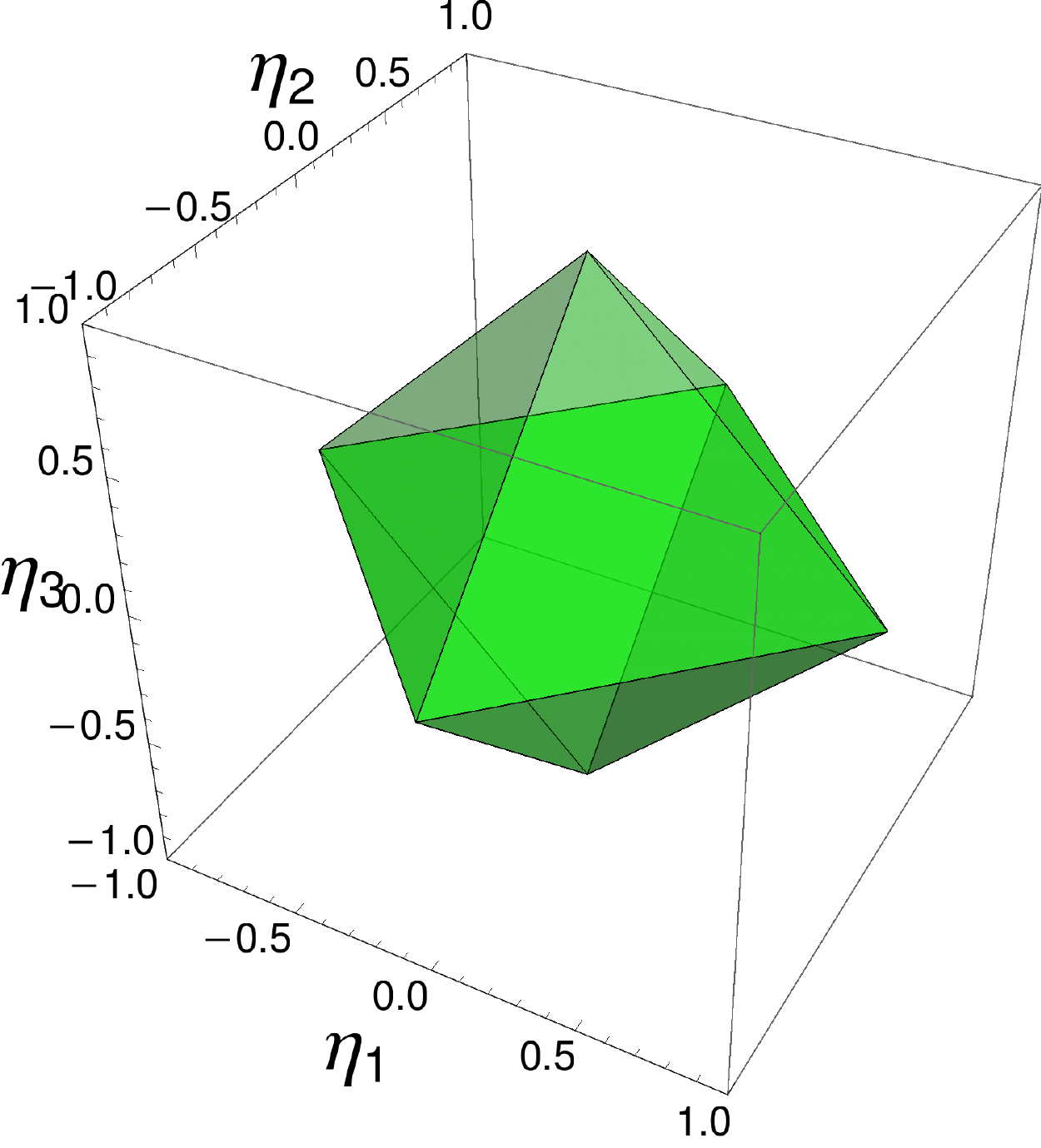}
\end{subfigure}
\caption{\textbf{Geometry of positive maps}: For positive trace
  preserving single-qubit maps, the allowed signed singular values lie
  within a cube $\mathcal{C}$ defined in \eqref{def:C}. Quantum
  channels corresponding to CPTP maps lie within the blue tetrahedron
  $\mathcal{T}_\text{CP}$ defined in \eqref{def:T_CP}, steering maps
  corresponding to CcPTP maps lie within the yellow tetrahedron
  $\mathcal{T}_\text{CcP}$ defined in \eqref{def:T_CcP}. The maps with
  SSV inside the intersection of $\mathcal{T}_\text{CP}$ and
  $\mathcal{T}_\text{CcP}$ (green octahedron) are called
  superpositive. These maps only produce classical correlations
corresponding to separable states or
  entanglement breaking channels, but can also be generated by
  mixtures of quantum correlations. }
\label{Fig:SSV}
\end{figure}

\subsection{Positive maps}
We have seen that a quantum channel is a CPTP map and that a steering
map is a CcPTP map. Both of them are necessarily positive maps. But
are there positive maps that are neither CcP nor CP? Or are there maps
that are even both?  This issue is nicely worked out in \cite[chapter
11]{Bengtsson2006}. We 
shortly review this  for unital 
qubit maps. Since we still deal with linear maps, it is
straightforward that also every 
unital positive one-qubit map can be
described by a $3\times 3$ correlation matrix. Hence we can
also analyze its SSV. 
The allowed SSV are inside the cube $\mathcal{C}$ defined by \cite[FIG.11.3]{Bengtsson2006} 
\begin{align}
\mathcal{C} \equiv \left\lbrace \boldsymbol{x}| -1\leq x_i \leq 1 \text{ for } i= 1,2,3\right\rbrace. \label{def:C}
\end{align}
This is illustrated in FIG.\ref{Fig:SSV}. Note again that we only
treat unital maps. \\
We see that there are  
positive maps which are neither CP nor CcP. 
According to the
\textit{St{\o}rmer-Woronowicz theorem} (see e.g.~\cite[p. 258]{Bengtsson2006})
 every positive qubit 
map is \textit{decomposable}, i.e.~it can be written as a convex
combination of a CP and a CcP map. Maps that are both CP and CcP are
called \textit{super positive} (SP). The set 
of allowed SSV of the correlation matrices of these maps forms an
octahedron (green region in FIG.\ref{Fig:SSV})  
given as 
\begin{align}
\mathcal{O}_\text{SP} = \text{Conv}\left(\{\pm \hat{e}_i | i \in {x,y,z}\}\right), \label{eq:octahedron}
\end{align} 
where $\hat{e}_i$ denotes the unit vector along the $i$-axis.
These correlations are generated by entanglement breaking quantum channels \cite{Ruskai2003} and 
steering maps based on separable states 
\cite{Jevtic}. 
When such \textit{classical} correlations are observed one cannot infer anything about the causal structure \cite[p.10 of supplementary information]{Ried2015}. \\
For higher dimensional systems things change. Already for three dimensional maps, i.e.~qutrit maps, there exist positive maps, that cannot be represented as a convex combination of a CP and a CcP map \cite[chapter 11.1]{Bengtsson2006}.
In the next section we discuss how much information about causal influences we can obtain  by looking  only at the SSV related to the correlations Alice and Bob can observe in a bipartite experiment.

\section{Causal explanation of unital positive maps} \label{sec:Results}
\subsection{Setting}
We now tackle the problem of causal inference in the two-qubit
scenario \cite{Ried2015}. 
The setting is as follows. An experimenter, Alice, sits in her
laboratory. She opens her door just long enough to obtain
a qubit in a (locally) completely mixed state and closes the door
again. She performs an projective measurement in any of the
Pauli-states, records her outcome, opens her door again and puts the
qubit in the now collapsed state outside. Apart from the qubit she has no way of
interacting with the environment. Some time later another
experimenter, Bob, opens the door of his laboratory and obtains a
qubit. Also he measures in the eigenbasis of
one of the Pauli matrices and 
 records the outcome.
They repeat this procedure a large (ideally: an infinite)
number of times. Then they meet and analyze their joint
measurement outcomes.  
These define the probabilities $P(a, b| j, i)$ for the outcomes
$a\in\{-1,1\}$ and $b\in\{-1,1\}$ of Alice's and Bob's measurements, 
given they measured in the eigenbasis of the $j$th and
$i$th Pauli matrix, respectively. For the marginals  we assume
$P(a|j,i) = \sum_b P(a,b|j,i) = 1/2 \,\, \forall a\in \{-1,1\}$ and 
accordingly for Bob. They are thus able to define a correlation matrix
$M$ with elements 
\begin{align}
M_{ij} =  2 P(b=1|j,i, a=1)-1 = \braket{\sigma_j \sigma_i},
\end{align}
where  $P(b=1|j,i, a=1)$ is the probability that Bob obtains outcome
$1$ when measuring the observable $\sigma_i$, conditioned on Alice's
measurement of $\sigma_j$ with outcome $1$, and
$\braket{\sigma_j\sigma_i}$ denotes the expectation value of the
product of Alice's $\sigma_j$ and Bob's $\sigma_i$ measurement
outcomes. \\
The correlation matrix defines a unique positive
trace preserving unital map $\mathcal{M}: \rho_A \mapsto \rho_B$.  
They 
are guaranteed one of the following three possibilities: 
either they measured the same qubit, which was 
propagated in terms of a unital quantum channel $\mathcal{E}$ from
Alice to Bob; 
or that they each measured one of the two 
qubits in a unital
bipartite state $\rho_{AB}$ acting as a common cause,
and hence the correlations where caused by 
the corresponding steering map $\mathcal{S}$; or that the map 
from $\rho_A$ to $\rho_B$ is a 
probabilistic mixture where with probability $p$ the steering map
$\mathcal{S}$ was realized and with probability $(1-p)$ the quantum
channel $\mathcal{E}$, that is 
 \begin{align}
\mathcal{M} = (1-p) \mathcal{E} + p \mathcal{S}, \label{eq:decomposition}
\end{align} with the ``causality parameter'' $p \in [0,1]$. 
\begin{figure}
\centering
\includegraphics[scale=.5]{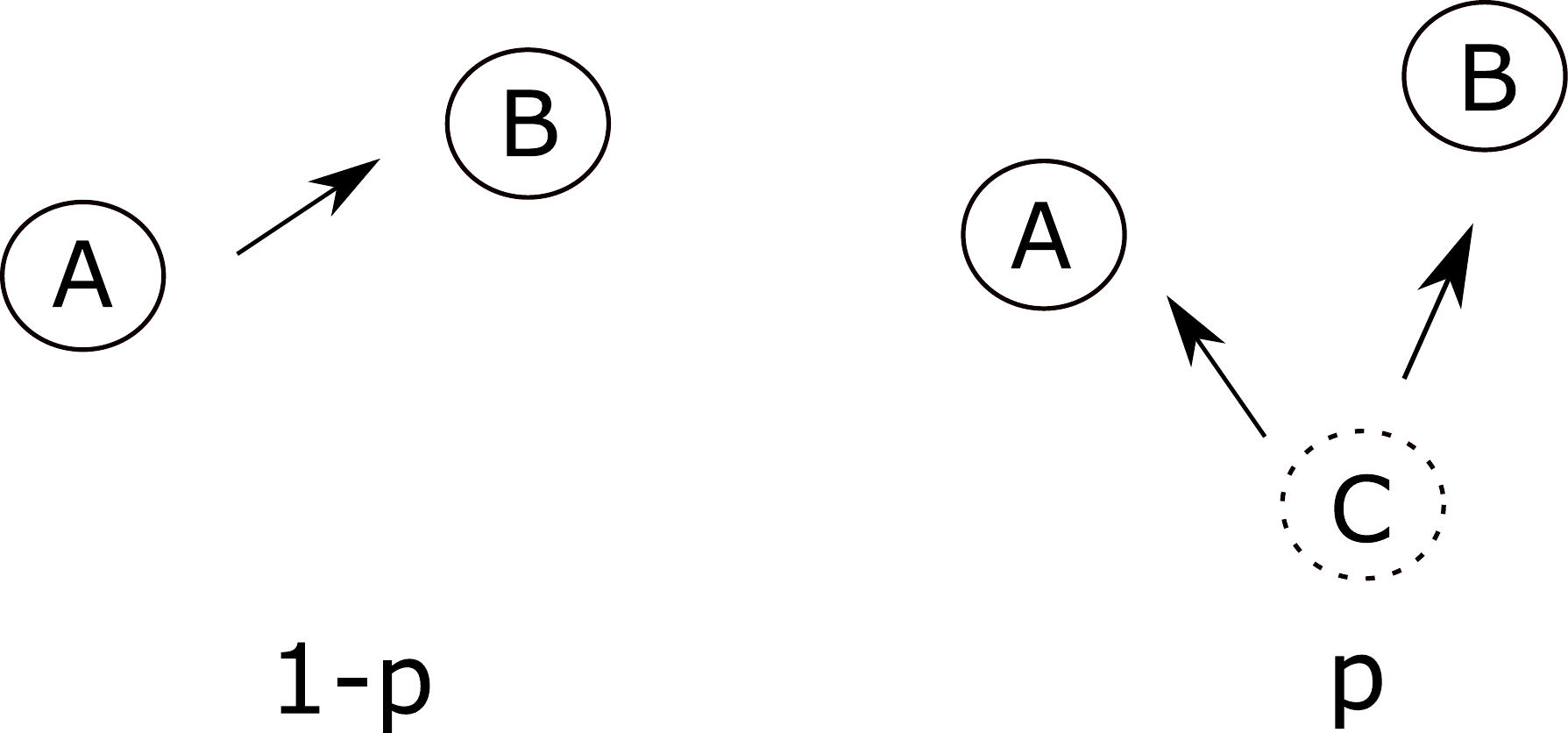}
\caption{\textbf{DAG:} The DAGs of our setting. On the left side with
  probability $(1-p)$ a quantum channel $\mathcal{E}$ is realized,
  causing correlations between Alice ($A$) and Bob ($B$). On the right
  side, occuring   
with probability $p$, the correlations are caused by an unobserved source $C$ that outputs the state $\rho_{AB}$ generating correlations through the steering map $\mathcal{S}$.}
\end{figure}
The task of Alice and Bob is now to find the true value of $p$ and
possibly also the nature of $\mathcal{S}$ and $\mathcal{E}$. In
general there 
does not exist a unique solution and in this case they want to find
the values of $p$ for which maps of the form \eqref{eq:decomposition}
explain the observed correlations.\\ 
As we mentioned in the previous section, every positive one qubit map
is decomposable, so 
a possible explanation always exists. The decomposition
\eqref{eq:decomposition} can be given a causal interpretation, where
$\mathcal{E}$ is considered to be a cause-effect explanation of the
correlations and $\mathcal{S}$ a common-cause.  
\\
In the following subsections we 
give bounds on the causality parameter $p$ and then consider some
extremal cases. In subsection \ref{sec:AdditionalConst}  
 we  
generalize a part of the work of Ried et al.~\cite{Ried2015} and see how additional
assumptions on the nature of $\mathcal{E}$ and $\mathcal{S}$ can lead
to a unique solution.

\subsection{Possible causal explanations}
\begin{defn}\textbf{$p$-causality/$p$-decomposability:} 
A single qubit unital positive trace preserving map $\mathcal{M}$ is called $p$-causal/$p$-decomposable with $p \in [0,1]$, if it can be written as
\begin{align}
\mathcal{M} = (1-p) \mathcal{E} + p \mathcal{S}, \label{def.pdec}
\end{align}
with $\mathcal{E}$ ($\mathcal{S}$) being a CPTP (CcPTP) unital 
qubit map. 
Eq.\eqref{def.pdec} is called a $p$-decomposition of
$\mathcal{M}$. 
\end{defn}
In the following let $M, E, S$ denote the 
correlation matrices of $\mathcal{M}, \mathcal{E}, \mathcal{S}$, and $\boldsymbol{\eta}^\mathcal{M}, \boldsymbol{\eta}^\mathcal{E}, \boldsymbol{\eta}^\mathcal{S}$ the SSV  of $M,E,S$, respectively. 
We first
investigate for a fixed $p$ what the
possible 
SSV of 
the correlation matrix of a map $\mathcal{M}$ are, such that
$\mathcal{M}$ is $p$-causal. This leads to the following theorem: 

\begin{theo}\textbf{Signed singular values of $p$-causal maps}\label{theo:p-causal}\\
Let $\mathcal{M}$ be a positive unital trace preserving qubit map with
associated SSV given by $\boldsymbol{\eta}^\mathcal{M}$. Let $p \in
[0,1]$ be fixed. Then the following statement holds:  
\begin{align}
\mathcal{M}\text{ is $p$-causal} \Leftrightarrow \boldsymbol{\eta}^\mathcal{M} \in \mathcal{C}_p,
\end{align}
where 
\begin{align}
\mathcal{C}_p = \text{Conv}\left( \left\{ (1-p) \boldsymbol{v}_i^\text{CP} + p \boldsymbol{v}_j^\text{CcP}| i,j \in\{1,2,3,4\}\right\} \right), \label{def:C_p}
\end{align}
where the vertices $\boldsymbol{v}_i^\text{CP}$ of CP maps are given in
\eqref{eq:CP_vertices}, and 
the vertices $\boldsymbol{v}_j^\text{CcP}$ of CcP maps in
\eqref{eq:CcP_vertices}.  
\end{theo} 
\begin{proof}
"$\Leftarrow$": From \eqref{def:C_p} we see that \begin{align*}
\boldsymbol{\eta}^\mathcal{M} \in \mathcal{C}_p \Leftrightarrow \exists \left(p_{ij} \geq 0, \;\sum_{i,j=1}^4 p_{ij} =1\right): \boldsymbol{\eta}^\mathcal{M}=\sum_{i,j=1}^4 p_{ij} \left((1-p) \boldsymbol{v}_i^\text{CP} + p \boldsymbol{v}_j^\text{CcP}\right).
\end{align*}
Now define $q_i \equiv \sum_j p_{ij}$ and $r_j \equiv \sum_i p_{ij}$. Clearly $q_i, r_j \geq 0$ and $\sum_i q_i = \sum_j r_j = 1$. 
We can then write
\begin{align*}
\boldsymbol{\eta}^\mathcal{M} = (1-p) \sum_i q_i \boldsymbol{v}_i^\text{CP} + p \sum_j r_j \boldsymbol{v}_j^\text{CcP} = (1-p)\, \boldsymbol{\eta}^\mathcal{E} + p \,\boldsymbol{\eta}^\mathcal{S},
\end{align*}
with $\boldsymbol{\eta}^\mathcal{E} \equiv \sum_i q_i 
\boldsymbol{v}_i^\text{CP} \in \mathcal{T}_\text{CP} $ and $
\boldsymbol{\eta}^\mathcal{S} \equiv \sum_j r_j 
\boldsymbol{v}_j^\text{CcP} \in \mathcal{T}_\text{CcP}$. We herewith
explicitly constructed a $p$-decomposition of $\mathcal{M}$ where the
correlation matrices of $\mathcal{E}$ and $\mathcal{S}$ have their
SSV-decomposition involving the same rotations as the
SSV-decomposition of the 
correlation matrix of $\mathcal{M}$.\\

"$\Rightarrow$": Let $p$ be fixed. Suppose that $\mathcal{E}$ and
$\mathcal{S}$ are both extremal maps, 
i.e.~
$\boldsymbol{\eta}^\mathcal{E}$ and $\boldsymbol{\eta}^\mathcal{S}$ are given by one
of the vertices defined in  \eqref{eq:CP_vertices} and
\eqref{eq:CcP_vertices}, respectively, and without loss of generality we assume that
these are $\boldsymbol{v}_1^\text{CP}$ and
$\boldsymbol{v}_1^\text{CcP}$ (this is justified as taking another
vertex leads to the same result). Define $A = (1-p) E$ and $B= p S$, where $A$ has SSV
$(1-p,1-p,1-p)$ and B has SSV $(-p,-p,-p)$. In the Appendix 
we prove theorem \ref{theo:SSV} that  
restricts the possible SSV of $A+B$.  
For our case it gives 
\begin{align*}
SSV(M) \in \mathcal{C}_p.
\end{align*}
Now suppose $\mathcal{E}$ and $\mathcal{S}$ are not extremal
maps. Since the SSV of those are simply convex combinations of the SSV
of the extremal maps, it follows that also for such maps the signed
singular values of $M$ lie within $\mathcal{C}_p$.
\end{proof} \\

 \begin{figure}
\begin{subfigure}{0.4\textwidth}
\centering
\includegraphics[scale=0.4]{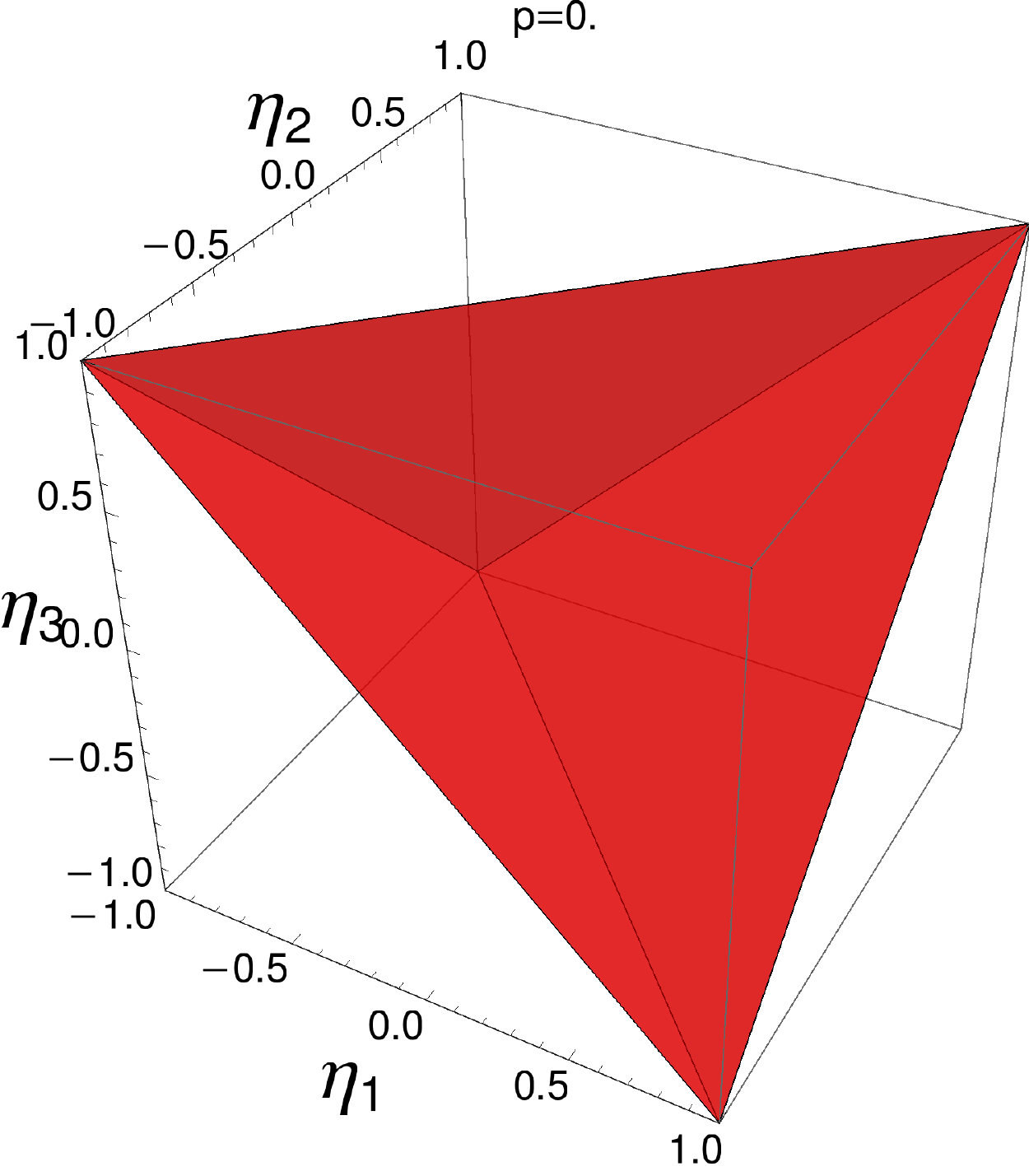}
\end{subfigure}
\begin{subfigure}{0.4\textwidth}
\centering
\includegraphics[scale=0.4]{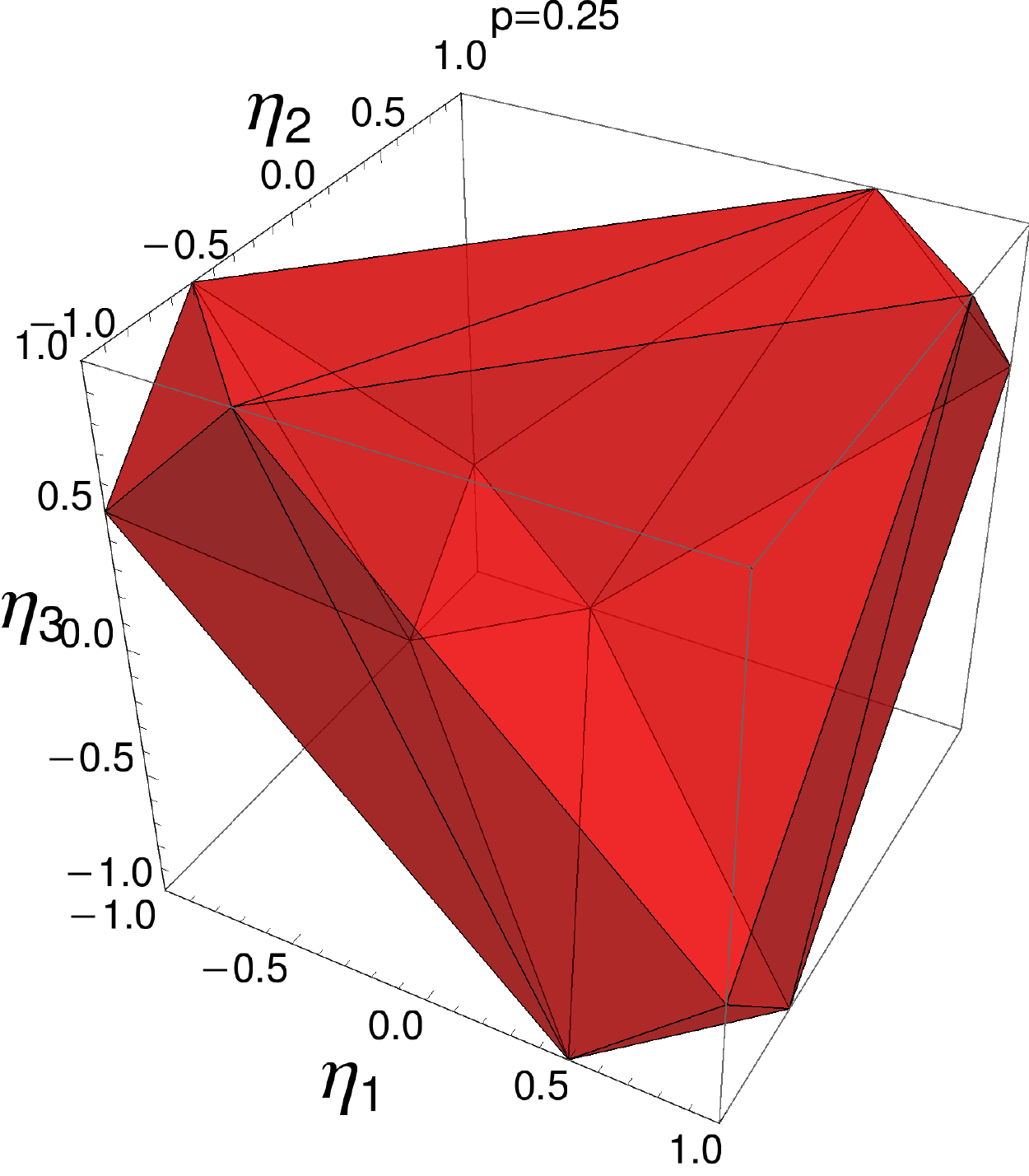}
\end{subfigure}
\begin{subfigure}{0.4\textwidth}
\centering
\includegraphics[scale=0.4]{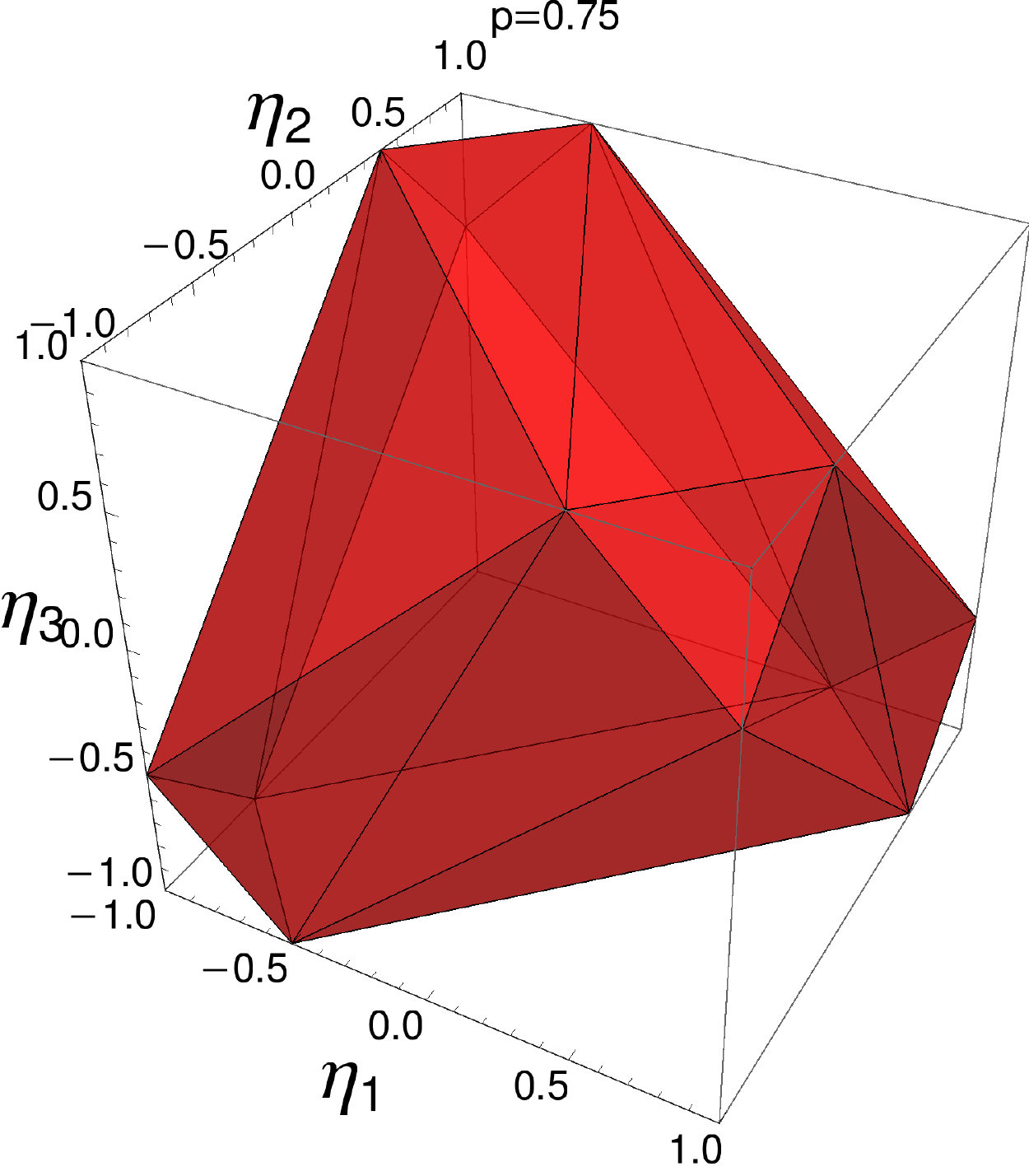}
\end{subfigure}
\begin{subfigure}{0.4\textwidth}
\centering
\includegraphics[scale=0.4]{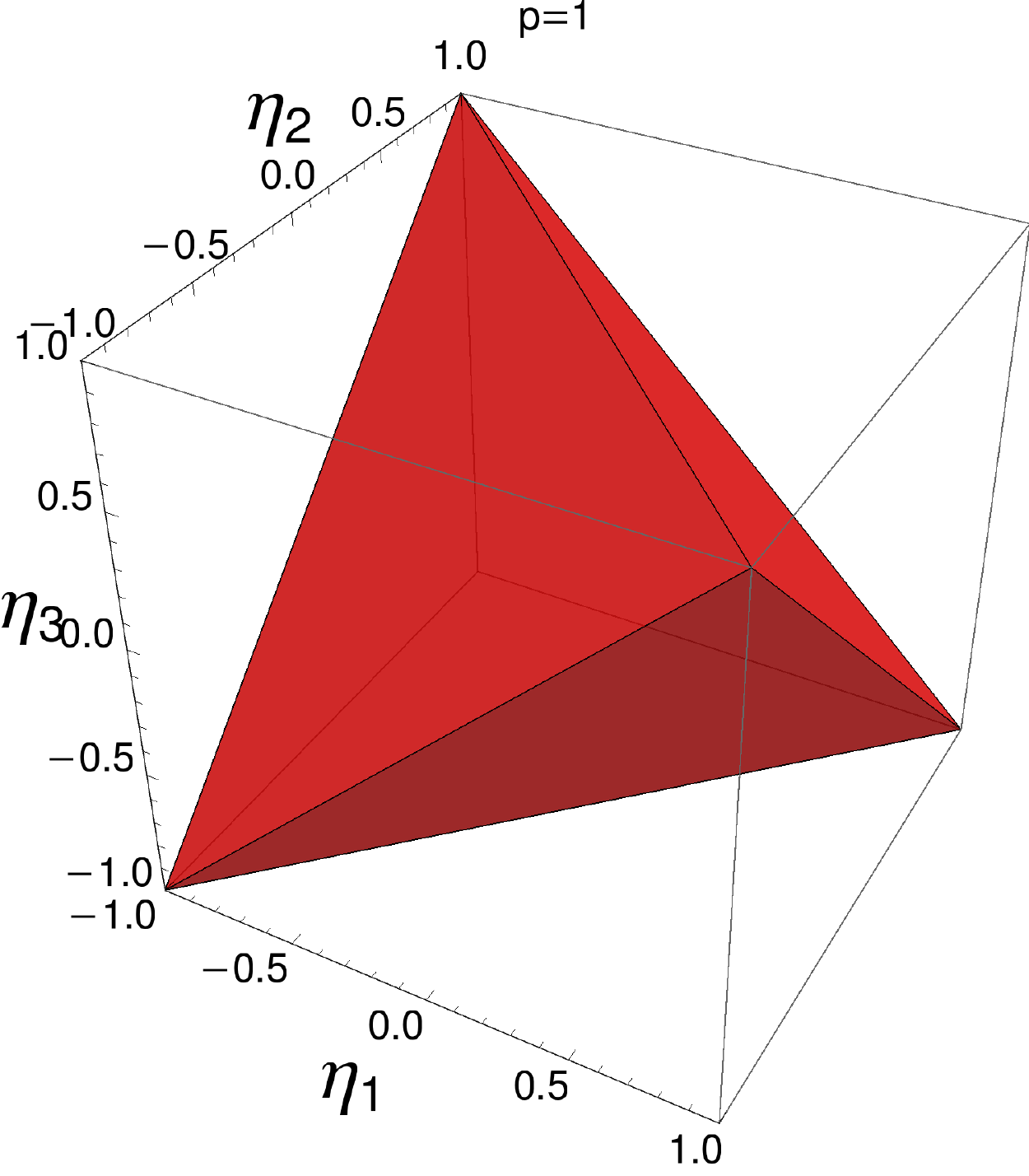}
\end{subfigure}
\caption{\textbf{Signed singular values of $p$-causal maps}: Set of
  attainable vectors of 
  signed singular values 
  associated with $\mathcal{M}$ in
  \eqref{eq:decomposition} for
  different values of $p$. 
By theorem \ref{theo:p-causal}, 
 for fixed $p$ there exists a CPTP map $\mathcal{E}$ and a CcPTP map
 $\mathcal{S}$ such that $\mathcal{M}$ is given by
 \eqref{eq:decomposition} if and only if the vector of signed singular
 values $\boldsymbol{\eta}^\mathcal{M}$ of 
 the correlation matrix of $\mathcal{M}$ 
  is in $\mathcal{C}_p$ defined in \eqref{def:C_p}.}
\label{Fig:SSV_mixture}
\end{figure}
We have seen that for a given value of $p$ the allowed SSV 
associated with a positive map $\mathcal{M}$ that is $p$-causal lie within $\mathcal{C}_p$ given in \eqref{def:C_p}.  We now turn the task around and go back to the causal inference scenario. 
Given a positive map $\mathcal{M}$ we want to tell if we can bound the
causality parameter $p$. We will do this based on the following
definition: 
\begin{defn}\textbf{Causal interval $I_\mathcal{M}$}:\\
For a given positive unital qubit map $\mathcal{M}$ we define the
interval of possible causal explanations (for short: the causal 
interval) $I_\mathcal{M}$, such that $\mathcal{M}$ is $p$-causal if
and only if $p\in I_\mathcal{M}$. 
\end{defn}
Since every qubit map is decomposable \cite[p.258]{Bengtsson2006} 
the causal interval is always non empty, $I_\mathcal{M} \neq \emptyset$.

\begin{theo}\label{theo:causal_interval}
Let $\mathcal{M}$ be a positive unital qubit map, with 
associated signed singular values $\boldsymbol{\eta}^\mathcal{M}$ (we assume $\eta^\mathcal{M}_i \geq 0$ for $i=1,2$). Then the causal interval of $\mathcal{M}$ is given by
\begin{align}
p_\text{max} &= \min\left(\frac{3 - \boldsymbol{\eta}^\mathcal{M}\cdot \boldsymbol{v}_1^\text{CP}}{2}, 1\right) \label{eq:pmax},\\
p_\text{min} &= \max\left(\frac{\boldsymbol{\eta}^\mathcal{M}\cdot \boldsymbol{v}_4^\text{CcP}-1}{2},0\right),\label{eq:pmin}
\end{align}
with $\boldsymbol{v}_1^\text{CP} = (1,1,1)^T$ ($\boldsymbol{v}_4^\text{CcP}  = (1,1,-1)^T$) defining a vertex of the CPTP (CcPTP) tetrahedron $\mathcal{T}_{CP}$ ($\mathcal{T}_{CcP}$).
\end{theo}
Note that the assumption $\eta^\mathcal{M}_i \geq 0$ for $i=1,2$ can
always be met, 
using the unitary freedom in the decomposition in the right way. \\
\begin{proof}
We show the theorem for $p_\text{max}$, the determination of
$p_\text{min}$ can be treated in an analogue way. \\ 
First we check if $\mathcal{M}$ is a CcPTP map, by checking if 
$\boldsymbol{\eta}^\mathcal{M} \in \mathcal{T}_\text{CcP}$.
If it is
CcPTP then $p_\text{max}=1$, trivially.\\ 
\begin{figure}
\centering
\includegraphics[scale=0.5]{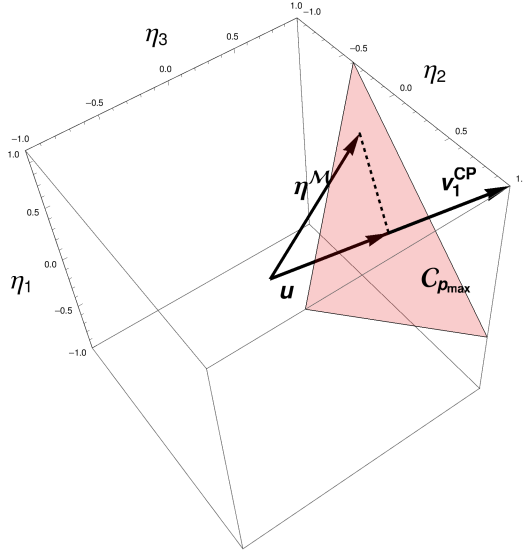}
\caption{\textbf{Sketch for proof of theorem
    \ref{theo:causal_interval}}: The value of $p_\text{max}$ is
  determined through the projection of $\boldsymbol{\eta}^\mathcal{M}$
  onto $\boldsymbol{v}_1^\text{CP}$, which is given by
  $\boldsymbol{u}$. The red triangle is one of the facets of
  $\mathcal{C}_{p_\text{max}}$.} 
\label{Fig:Sketch}
\end{figure}
Now suppose it is not CcPTP. $p_\text{max}$ is then given such that
$\boldsymbol{\eta}^\mathcal{M} \in \mathcal{C}_{p_\text{max}}$ but
$\boldsymbol{\eta}^\mathcal{M} \notin \mathcal{C}_{p'}$ with $p' \in
(p_\text{max},1]$.
This implies that $\boldsymbol{\eta}^\mathcal{M}$
lies on the surface of $\mathcal{C}_{p_\text{max}}$. Since we assumed
$\eta^\mathcal{M}_i \geq 0$ for $i=1,2$,  the critical facet of
$\mathcal{C}_{p_\text{max}}$ is the one which is perpendicular to
$\boldsymbol{v}^\text{CP}_1$ and has the vertices
$(1,1,1-2p_\text{max})^T,(1,1-2p_\text{max},1)^T,
(1-2p_\text{max},1,1)^T$ (see FIG.\ref{Fig:Sketch}). Since this facet
is perpendicular to $\boldsymbol{v}^\text{CP}_1$,
$\boldsymbol{\eta}^\mathcal{M}$ lies on this facet if its projection
onto $\boldsymbol{v}_1^\text{CP}$ equals the vector pointing from the
origin to the intersection of the facet and
$\boldsymbol{v}^\text{CP}_1$, given as
$\boldsymbol{u}\equiv(1-(2/3)\,p_\text{max})\boldsymbol{v}^\text{CP}_1$,
see Fig.\ref{Fig:Sketch}. Hence we get the following equation  
\begin{align}
&\frac{1}{3}\boldsymbol{v}^\text{CP}_1 (\boldsymbol{v}^\text{CP}_1
  \cdot \boldsymbol{\eta}^\mathcal{M}) \overset{!}{=} \boldsymbol{u} 
\\
&\Leftrightarrow \boldsymbol{v}^\text{CP}_1 \cdot \boldsymbol{\eta}^\mathcal{M} = 3 - 2p_\text{max}\\
&\Leftrightarrow p_\text{max} = \frac{3 - \boldsymbol{\eta}^\mathcal{M}\cdot \boldsymbol{v}^\text{CP}_1}{2}. 
\end{align}

\end{proof}
\subsection{Extremal cases}\label{sec:extremal} 
In the previous section we found the general form of the causal
interval $I_\mathcal{M}$ for an observed map $\mathcal{M}$. We 
now analyze the extremal cases where the interval reduces to a single value or on the other hand the interval is given as $I_\mathcal{M}=[0,1]$.\\
As already noted in \cite[Table 1.]{Ried2015} there are extremal cases
that allow for a complete solution of the problem even without any
additional constraints. This is the case if 
$\boldsymbol{\eta}^\mathcal{M}$ equals one of the vertices of the cube of positive maps, see
Fig. \ref{Fig:SSV}. The solution is then either $p=0$ (pure
cause-effect) if the SSV are all positive or exactly two are negative
or $p=1$ (pure common-cause) if the SSV are all negative or exactly
one positive. The exact reconstruction of $\mathcal{E}$ or
$\mathcal{S}$ in this cases is trivial.\\ 
Interestingly, with theorem \ref{theo:causal_interval} we can show that
\textit{every} point on the edges of the cube 
$\mathcal{C}$ defined in \eqref{def:C} gives us a unique solution
without additional constraints: \\
\begin{proof}
Let $\mathcal{M}$ be a  positive map and $M$ be the corresponding
correlation matrix with $M = R_1 \eta^\mathcal{M} R_2$ where
$\eta^\mathcal{M} = \diag (\boldsymbol{\eta}^\mathcal{M})$ with the
signed singular values $\boldsymbol{\eta}^\mathcal{M} = (1,1,1-2p)^T,
\; p \in [0,1]$, and two rotations $R_1,R_2 \in SO(3)$. Due to the freedom in $R_1$ and $R_2$ this
describes all maps with corresponding vector of SSV on one of the edges of the cube $\mathcal{C}$ defined in
\eqref{def:C}. According to theorem \ref{theo:causal_interval} we find  
\begin{align}
p_\text{max} &= \min\left(\frac{3 - \boldsymbol{\eta}^\mathcal{M}\cdot
               \boldsymbol{v}_1}{2}, 1\right) = \frac{3- \left(2+(1-2p)\right)}{2}
               = p,\\ 
p_\text{min} &= \max\left(\frac{\boldsymbol{\eta}^\mathcal{M}\cdot
               \boldsymbol{v}_4^\text{CcP}-1}{2},0\right)=
               \frac{2-(1-2p)-1}{2} = p. 
\end{align}
By theorem \ref{theo:SSV} it follows, that the maps $\mathcal{E}$ and
$\mathcal{S}$ in the decomposition \eqref{eq:decomposition}
necessarily correspond to extremal points in $\mathcal{T}_\text{CP}$
and $\mathcal{T}_\text{CcP}$ defined in \eqref{def:T_CP} and
\eqref{def:T_CcP} (unitary channel and maximally entangled state). It
is then obvious that \begin{align} 
 M = R_1\left((1-p)\diag(\boldsymbol{v}^\text{CP}_1) + p\,\diag(
                       \boldsymbol{v}^\text{CcP}_4)\right) R_2 
\end{align}
is the only possible solution.
\end{proof}

In the other extreme case, 
 if the map $\mathcal{M}$ is superpositive,
i.e.~CP and CcP (see Figure \ref{Fig:SSV}), it could be explained by a
pure CPTP, a pure CcPTP map, or any convex 
combination of those
two. Therefore 
one cannot give any restrictions of possible values of
$p$  \cite[III.E of supplementary information]{Ried2015}.\\ 
\begin{proof}
Let $\mathcal{M}$ be a superpositive map. 
There exists a SSV decomposition 
of its correlation matrix for which $\boldsymbol{\eta}^\mathcal{M} \in
\mathcal{O}_\text{SP}$, defined in \eqref{eq:octahedron}, and for which
$\eta^\mathcal{M}_i \geq 0$ for $i=1,2$. Hence we can write
$\boldsymbol{\eta}^\mathcal{M} = p_1 \hat{e}_x + p_2 \hat{e}_y + p_3
\hat{e}_z + p_4 (-\hat{e}_z)$, with $\sum_i p_i =1$. The scalar
product of each component of $\boldsymbol{\eta}^\mathcal{M} $ with
$\boldsymbol{v}^\text{CP}_1 = (1,1,1)^T$ is upper bounded by 1. Hence
we have $\boldsymbol{\eta}^\mathcal{M} \cdot
\boldsymbol{v}^\text{CP}_1 \leq 1$ and with that eq. \eqref{eq:pmax}
evaluates to $p_\text{max} = 1$. Analogously one finds $p_\text{min}
=0$. 
\end{proof}

\subsection{Additional assumptions / Causal inference with constrained classical correlations} \label{sec:AdditionalConst}
So far we only assumed that our data is generated by a unital channel
and a unital state (a state whose local partitions are completely
mixed). We have seen that in some extreme cases a unique solution to
the problem can be found. Ried et al.~showed that one
can always find a unique solution for $p$ 
if one restricts the channel
to unitary channels and the bipartite states to maximally entangled
pure states \cite{Ried2015}. Furthermore, it is then possible to reconstruct the
channel and the state up to binary ambiguity, meaning there are two
explanations leading to the same observed correlations. 
The ellipsoids 
associated with unitary channels and maximally entangled states  are
spheres with unit radius and the SSV 
of 
their correlation matrices correspond to the vertices of
$\mathcal{T}_\text{CP}$ and $\mathcal{T}_\text{CcP}$ respectively . \\ 
In the following we investigate this scenario again, but add a
known amount of noise in the channel or in the bipartite
state. For the channel this is done by mixing the unitary evolution
with a completely depolarizing channel \cite{Nielsen2009}. The
completely depolarizing channel maps every Bloch-vector to the origin,
$\rho \mapsto \frac{\mathbb{1}}{2}$  and hence is represented by the
zero matrix. The 
ellipsoid associated with the
 mixture of a completely depolarizing channel with a
unitary channel thus results in a shrinked sphere. 
For strong enough noise the result eventually becomes an entanglement
breaking channel, which only produces ``classical'' correlations
\cite{Ruskai2003}. Due to the unitary freedom compared to standard
depolarizing channels, 
we call these channels \textit{generalized depolarizing
  channel}. For the state we mix a pure maximally entangled state with
the completely mixed state, whose correlation matrix is given by the zero-matrix. 
We call the state a \textit{generalized Werner state}, in the
sense that instead of a convex combination of a singlet and a
completely mixed state \cite{Werner1989} we allow the convex
combination of an arbitrary maximally entangled state with the
completely mixed state. States at a certain threshold of noise become
separable and the correlations become ``classical'' \cite{Jevtic}. We
will then see that even when confronted with purely classical
correlations, if we have enough a-priori-knowledge about the data
generation, \textit{i.e.}~we know the amount of noise, we can still
find a solution analogous to \cite{Ried2015}, in the sense of
determining uniquely the parameter $p$,    
and the channel and the state up to binary
ambiguity\footnote{Strictly speaking, only for $p\neq 1/2$ one can always
determine the unitary and the state. For $p=1/2$ there is an infinite number of
channels and states (all those where every point is diametrically
opposed for the unitary channel and the state.), for which the
ellipsoid reduces to a single point, and hence the correlation matrix is the zero matrix. The parameter $p=1/2$ can then be restored but not the unitary and the state.}.
We will first keep the unitary channel and start with a generalized
Werner-state and show how one can recreate the scenario of Ried et
al. Then we will add the noise in the channel. 
 
\subsubsection{Solution of the causal inference problem using  generalized Werner states}
The  analysis follows closely in spirit section III.D in the supplementary information in
\cite{Ried2015}.
 We start again with equation \eqref{eq:decomposition} and assume that the steering map $\mathcal{S}$ is generated by a shared generalized Werner state $\rho_{AB} = \epsilon \frac{\mathbb{1}}{4} + (1-\epsilon) \ket{\psi}\bra{\psi}$, where the parameter $\epsilon \in [0,1)$ is known and fixed in advance and $\ket{\psi}$ is an unknown maximally entangled pure state. The map $\mathcal{E}$ is generated by an unknown unitary channel $U$.\\
Since $\epsilon$ is fixed, the class of allowed
explanations 
is completely defined up to unitary freedom in the channel and in the state. Hence the number of free parameters is the same as in the case considered in \cite{Ried2015}, which coincides with the case $\epsilon=0$. For
$\epsilon > 2/3$ the state $\rho_{AB}$ becomes separable, i.e.~is not
entangled anymore, see 
\cite{Werner1989} and Fig.5 in the supplementary information of \cite{Jevtic}. But
the reconstruction works independently of $\epsilon$.  Hence, we see
here that the possibility of reconstruction hinges not on the
entanglement in $\rho_{AB}$ but on the prior knowledge we have about
$\rho_{AB}$.\\

The correlation matrix 
corresponding
to the generalized Werner-state is simply the one of a maximally
entangled state shrinked by a factor $1-\epsilon$ and will thus be
denoted $(1-\epsilon) S$, where $S$ is the 
correlation matrix
corresponding to a maximally entangled state. Thus in our scenario the 
information Alice 
and Bob obtain characterizes the matrix \begin{align}
M = p(1-\epsilon) S +(1-p) E. \label{eq:Mwerner}
\end{align}
The ellipsoid is described by the eigenvalues and -vectors of
$MM^T$. The eigenvectors correspond to the direction of the semi
axes and the squareroots of the eigenvalues are their lengths. There
is one degenerate pair and another single one. The eigenvector
corresponding to the non-degenerate semi axis is parallel to
$\boldsymbol{n}$ which is defined as the axis on which the images of
$S$ and $E$ are diametrically opposed.  
 Hence the length of this semi axis is $l_1 =|1-p -
 p(1-\epsilon)|$. Furthermore we have 
\begin{align*}
&\text{sign}(\det M) =\text{sign} (1-p - p(1-\epsilon)),
\end{align*}
if $l_1>0$ 
and $\det M = 0$ if $l_1=0$. 
 Thus if we calculate the length of this semi axis we can already determine the causality parameter $p$ as
\begin{align}
l_1 = |1-2p +p\epsilon| \Leftrightarrow p = \frac{1\mp l_1}{2-\epsilon}, \label{eq:p2}
\end{align}
where the ambiguity is solved by considering the sign of $\det M$.\\
Now that we have $p$ and $\epsilon$ at hand we can define a 
new map with correlation matrix
\begin{align}
M' \equiv \frac{1}{1-p\epsilon} M = \frac{p(1-\epsilon)}{1-p\epsilon}S + \frac{1-p}{1-p\epsilon} E \equiv p' S + (1-p') E, \label{eq:T'}
\end{align}
where we defined
\begin{align}
&p' \equiv  \frac{p(1-\epsilon)}{1-p\epsilon},\\
& 1- p' = 1-  \frac{p(1-\epsilon)}{1-p\epsilon} = \frac{1-p\epsilon - p(1-\epsilon}{1-p\epsilon} = \frac{1-p}{1-p\epsilon}.
\end{align}
The 
properties of the ellipsoid can also be found in the SSV decomposition
of the 
correlation matrix
\begin{align}
M = R_1 DR_2, \qquad \text{with } D=\diag(\boldsymbol{\eta}^\mathcal{M}) \text{ and } \; R_1,R_2 \in SO(3).
\end{align} 
The absolute values of the entries of $\boldsymbol{\eta}^\mathcal{M}$ equal the
lengths of the semi axes of the ellipsoid and we  choose $R_1$ and
$R_2$ such that 
$\eta^\mathcal{M}_1= \eta^\mathcal{M}_2$. The axis on which the images
of $S$ and $E$ are diametrically opposed is then given by the last
column of $R_1$, i.e.~$\hat{n} = R_1 \hat{e}_3$. The length of this
axis is  $l_1 = |\eta^\mathcal{M}_3|$.\\ 
In \eqref{eq:T'} the promise is given that $S$ is the 
correlation matrix 
of a
maximally entangled state and that $E$ is the 
correlation matrix
of a unitary
channel. The reconstruction of those is extensively studied in the
supplementary information of \cite{Ried2015}.  With the method
presented there we find the value of $p'$ and can restore the
correlation matrices corresponding to $U$ and $\ket{\psi}$ up to a binary
ambiguity, and hence solve the causal inference problem. We review
this in terms of SSV and discuss where the binary ambiguity arises. \\ 
Starting from the l.h.s.~of  \eqref{eq:T'} the goal is to determine
$p', S,$ and $M$ on the r.h.s. Consider the SSV decomposition of
the  correlation matrix
\begin{align}
M' = R'_1 D'R'_2, \qquad \text{with } D'=\diag(\boldsymbol{\eta}^\mathcal{M'}) \text{ and } R_1',R_2'\in SO(3).
\end{align} 
The absolute values of the entries of $\boldsymbol{\eta}^\mathcal{M'}$ equal the lengths of the semi axes of the ellipsoid and we choose $R'_1,R'_2$ s.t.~$\eta_1^\mathcal{M'}=\eta_2^\mathcal{M'} $. The axis on which the images of $S$ and $E$ are diametrically opposed is then given by the last column of $R'_1$, i.e.~$\hat{n}' = R'_1 \hat{e}_3$. The length of this axis is $l'_1 = |\eta_3^\mathcal{M'}|$. 
However, the direction of $\hat{n}'$, depending on the choice of
$R'_1$ and $R'_2$, cannot  be determined uniquely and  allows two
possible solutions $\pm\hat{n}'$. The  parameter $p'$ is determined by
the length $l_1'$ and can be calculated as 
\begin{align}
p' = \frac{1 - \left(\text{sign}\det(M') \right)l_1'}{2},
\end{align}
and if $\det(M') = 0$ we have $p' = 1/2$. If $p'=0$ or $p'=1$ the reconstruction is trivial (of course in these cases one cannot reconstruct $S$ or $E$, respectively). If $p'\in (0,1)$, we can define \cite{Ried2015}  
\begin{align}
r' &= |\eta_1^\mathcal{M'}|,\\
\gamma' &= 2 \arcsin\left( \sqrt{\frac{1-r'^2}{4(p'-{p'}^2)}}\right),\\
\gamma_2' &= \arccos\left(\frac{1+r'^2- \left[2p'\sin\frac{\gamma'}{2}\right]}{2r'}\right). 
\end{align}
The reconstruction of the 
correlation matrices 
$S$ and $E$ can then be done, c.f. eq. (58) and (59) in the supplementary information of \cite{Ried2015}:
\begin{align}
E &= R_{\hat{n}', \gamma'_2} S_{\perp\hat{n}', 1/r'} S_{\hat{n}', 1/(1-2p')} M', \label{eq:ERec}\\
S &= R_{\hat{n}', -\gamma'+\gamma'_2} S_{\perp\hat{n}', 1/r'} S_{\hat{n}', 1/(2p'-1)} M', \label{eq:SRec}
\end{align}
where $R_{\hat{n}, \alpha}$ indicates a rotation about axis $\hat{n}$ with rotation angle $\alpha$, $ S_{\hat{n}', 1/(1-2p')}$ a scaling
along $\hat{n}'$ by a factor $1/(1-2p')$ and $S_{\perp\hat{n}', 1/r'}
$ a scaling of the plane perpendicular to $\hat{n}'$ by a factor
$1/r'$.
From \eqref{eq:ERec} and \eqref{eq:SRec} we see that a
reconstruction of $E$ and $S$ is not possible if $p'=1/2$. \\
Let us
summarize what we can infer 
about the causation of $M$ given in \eqref{eq:Mwerner}:
\begin{itemize}
\item 
The causality parameter $p$ can be determined uniquely in all cases, see eq.\eqref{eq:p2}.
\item If $r'=0$ or $p'=1/2$ then $S$ and $E$ cannot be determined,
\item else we can determine two sets of solutions for $E$ and $S$ given by \eqref{eq:ERec} and \eqref{eq:SRec}, distinguished by the choice of direction of $\hat{n}'$.
\end{itemize}
On the other hand, if we do not have prior knowledge of $\epsilon$, then in general we cannot determine $p$ with \eqref{eq:p2}. This ambiguity can easily be illustrated by looking at an example:\\
Take 
$U = \sigma_x$ and $\ket{\psi} = \frac{\ket{00} -\ket{11}}{\sqrt{2}}$. We then have:
\begin{align*}
E = \text{diag}(1,-1,-1), \qquad S = \text {diag}(-1,1,1).
\end{align*}
Combining this for arbitrary $\epsilon$ and $p$ gives 
\begin{align*}
M = \text{diag}(1-p\epsilon, -(1-p\epsilon), -(1-p\epsilon)).
\end{align*}
Hence for all values of the parameters where $p\epsilon = \text{cons.}$, the measurement statistics for Alice and Bob are exactly the same and there is no way to distinguish different pairs of values.\\

Analogously to using a generalized Werner state for the steering map, we can also use a generalized depolarizing channel. Then, with prior knowledge of the amount of noise, we can still find a complete solution even though the resulting channel might be entanglement breaking.

\subsubsection{Generalized depolarizing channel and generalized Werner state }
We shall now consider the case where both the channel as well as the
state are mixed with a known amount of noise. Therefore we take $S'
= (1-\epsilon_c) S$ 
for a generalized Werner state (thus $S$ 
corresponds again to a rotated and inverted Bloch-sphere) and $E' =
(1-\epsilon_e) E$ 
for a generalized depolarizing channel. We again assume $\epsilon_e\,
\in (0,1)$ and $\epsilon_c\, \in (0,1)$ to be known. We then have 
\begin{align}
M = (1-p)(1-\epsilon_e) E + p (1-\epsilon_c) S.
\end{align}
The reconstruction works as follows. Without loss of generality we assume $\epsilon_e \leq \epsilon_c$ (in the other case we just have to make the reconstruction discussed in the previous subsection for the entanglement breaking channel and not for the Werner-state). The only thing we have to do is to divide by $(1-\epsilon_e)$ to restore the problem of the previous section
\begin{align*}
M' = \frac{M}{1-\epsilon_e} = (1-p) E + p \frac{1-\epsilon_c}{1-\epsilon_e} S \equiv  (1-p) E + p (1-\epsilon) S,
\end{align*}
with $1-\epsilon \equiv \frac{1-\epsilon_c}{1-\epsilon_e}$. The rest can then be solved as in the previous subsection.\\
Again we remark that nothing changes if we have $\epsilon_c > 2/3$ and
$\epsilon_e > 2/3$ even though at that transition the states become separable and the channels entanglement-breaking, respectively.

\section{Discussion} \label{sec:discussion}
In this work we extended the results initially found by Ried et
al.~\cite{Ried2015}. We introduced an active and a passive quantum-observational scheme as analogies to the classical observational scheme. The passive quantum-observational scheme does not allow for an advantage over classical casual inference. In the active quantum observational scheme Alice and Bob can freely choose their measurement bases, which in principle allows for signaling. However, we investigated the quantum advantage over classical causal inference in a scenario where signaling is not possible in the active quantum observation scheme, as Alice' incoming state is completely mixed.\\

 We 
showed how the geometry of the set of signed singular values (SSV) of 
correlation matrices representing 
 positive maps of the density 
 operator $\rho_A\mapsto\rho_B$
determines the possibility to reconstruct the causal structure linking
$\rho_A$ and $\rho_B$.  We showed that there are more cases than
previously known for which a complete solution of the causal inference
problem can be found without additional constraints, namely all
correlations created by maps whose signed singular values of the 
correlation matrix 
lie on the
edges of the cube of positive maps $\mathcal{C}$ defined in
\eqref{def:C}. A necessary and sufficient  condition for this is
that the state is maximally entangled, that the channel is unitary,
and that the corresponding 
correlation matrices 
have a SSV decomposition 
involving the same rotations. \\ 
For correlations guaranteed to be
produced by a
mixture of a unital channel and a unital bipartite state,  we
quantified the quantum advantage by giving the intervals for possible
values of the causality parameter $p$. Here, in order to constrain $p$,
and hence have an advantage over classical causal inference, it is 
necessary that the correlations were caused by an entangled
state and/or an entanglement preserving channel. This is because
correlations caused by any mixture of a separable state and an
entanglement breaking  quantum channel always describe super-positive
maps. According to theorem \ref{theo:causal_interval} the causal
interval for any super-positive map 
$\mathcal{M}$ is $I_\mathcal{M} = [0,1]$.  Hence, super-positive maps do not allow any causal inference.  \\
Things change when we further strengthen the assumptions on the data generating
processes and  
allow only unitary freedom in the state, corresponding to a
generalized Werner state with given degree of noise $\epsilon_c$, or unitary freedom in the channel,
corresponding to a generalized depolarizing channel with given degree of noise $\epsilon_e$. 
We showed that in this scenario the causality parameter $p$ can always be
uniquely determined and in most cases the state and the channel can be
reconstructed up to binary ambiguity. For $\epsilon_c >
2/3$ the state becomes separable and for $\epsilon_e > 2/3$ the channel entanglement breaking
but still causal inference is feasible. Therefore entanglement and
entanglement preservation are not a necessary condition in this
scenario. The assumptions 
on the data generating processes, i.e. a-priori knowledge of
$\epsilon_c$ and $\epsilon_e$,  
are strong enough, such that even  correlations corresponding to super-positive maps reveal the
underlying causal structure.\\ 

\newpage

\section{Appendix} \label{sec:Appendix}
\subsection*{Signed singular values of sums of matrices}
Let $A$ be a $n\times n$ real matrix. A possible \textit{singular value decomposition} (SVD) of $A$ is given as \begin{align}
A = O_1 D O_2, \label{def: App_ASV}
\end{align}
 where $O_{1,2}$ are orthogonal matrices ($O_i O_i^T = \mathbb{1}$) and $D$ is a positive semi-definite diagonal matrix $D = \diag(\sigma_1^A,...,\sigma_n^A)$, with $\sigma_i^A$ called the (absolute) \textit{singular values} (SV) of $A$. 
 The matrices in \eqref{def: App_ASV} are not uniquely defined and all
 possible permutations of the singular values on the diagonal of $D$
 are possible for different orthogonal matrices $O_1$ and $O_2$. We use this
 freedom to write the SV in \textit{canonical} order, 
$\sigma_1^A \geq \sigma_2^A \geq ...\geq \sigma_n^A$. \\
\textit{Example:} We give two different SVDs of a $3\times3$ matrix $B$ 
 \begin{align*}
 B \equiv \begin{pmatrix}
 -1 & 0&0\\ 0&0&-3\\0&2&0
 \end{pmatrix} &=\begin{pmatrix}
 1&0&0\\0&0&-1\\0&1&0
\end{pmatrix}  \begin{pmatrix}
 1&0&0\\0&2&0\\0&0&3
 \end{pmatrix} 
 \begin{pmatrix}
 -1&0&0\\0&1&0\\0&0&1
 \end{pmatrix} \\
 &= \begin{pmatrix}
 0&-1&0\\ 0&0&1\\-1&0&0
 \end{pmatrix} \begin{pmatrix}
 3&0&0\\0&2&0\\0&0&1
 \end{pmatrix}\begin{pmatrix}
 0&0&1\\0&1&0\\1&0&0
 \end{pmatrix}.
 \end{align*}
 The last decomposition gives the singular values of $B$ in canonical order $\sigma_1^B =3,\sigma_2^B = 2$ and $\sigma_3^B = 1$.\\
Next we call 
\begin{align}
A = R_1 D' R_2 \label{def:App_SSV}
\end{align}
 the \textit{signed singular value decomposition} (also called real
 singular values \cite{Amir-Moez1958}) of $A$, where $R_i \in SO(n)$ are
 orthogonal matrices with determinant equal to one. In the $3\times 3$ scenario these correspond to proper rotations
 in $\mathbb{R}^3$. The diagonal matrix $D'$
 contains the \textit{signed singular values} (SSV) of A. The SSV have
 the same absolute values as the SV but additionally can have  negative
 signs. Concretely, the freedom in choosing $R_1$ and $R_2$ allows
one
 to get any permutations of the SV on the diagonal of $D'$ together
 with an even or odd number of minus signs, depending on whether $A$
 has positive or negative determinant, respectively. If at least one
 singular value equals 0, the number of signs becomes completely
 arbitrary. Using the same matrix $B$ as above we 
 give two
 different signed singular value decompositions as an
 \textit{example:} 
 \begin{align}
  B \equiv \begin{pmatrix}
 -1 & 0&0\\ 0&0&-3\\0&2&0
 \end{pmatrix} &=\begin{pmatrix}
 1&0&0\\0&0&-1\\0&1&0
\end{pmatrix}  \begin{pmatrix}
 -1&0&0\\0&2&0\\0&0&3
 \end{pmatrix} 
 \begin{pmatrix}
 1&0&0\\0&1&0\\0&0&1
 \end{pmatrix} \\
 &= \begin{pmatrix}
 0&-1&0\\ 0&0&1\\-1&0&0
 \end{pmatrix} \begin{pmatrix}
 3&0&0\\0&2&0\\0&0&-1
 \end{pmatrix}\begin{pmatrix}
 0&0&1\\0&1&0\\-1&0&0
 \end{pmatrix}. \label{eq:canonicalSSV}
 \end{align}
 For the SSV decomposition we 
 define a \textit{canonical} order  
with the absolute values of the singular values sorted
in decreasing order and only a negative sign on the last entry if the
matrix has negative determinant, as in \eqref{eq:canonicalSSV}. 
The rotational freedom in
\eqref{def:App_SSV} allows for arbitrary permutations of the order of
singular values and addition of any even number of minus signs.\\ 
 Confusion may arise since for example an $\mathbb{R}^3$ permutation matrix corresponding to a permutation of exactly two coordinates has determinant -1, so why would it be allowed? The point is, that we not only want to permute elements of a vector, but the diagonal elements of a matrix. We illustrate that by permuting two components of \textit{i)} a vector and \textit{ii)} a diagonal matrix. 
 \begin{align}
 &P_{yz} \equiv \begin{pmatrix} 1&0&0\\0&0&1\\0&1&0  \end{pmatrix}, \qquad \det P_{yz} = -1,\\
 & \textit{i) }P_{yz}\cdot \begin{pmatrix}
 a\\b\\c  \end{pmatrix} = \begin{pmatrix}
 a\\c\\b
 \end{pmatrix},\\
 &\textit{ii) }P_{yz} \cdot\text{diag}(a,b,c) \cdot P_{yz} = \text{diag}(a,c,b) = (-P_{yz}) \, \text{diag}(a,b,c)\, (-P_{yz}).
 \end{align}
 I.e.~as $-P_{yz} = R_{\hat{\boldsymbol{x}}}(\pi/2)\cdot
 R_{\hat{\boldsymbol{y}}}(\pi) $ the effect of permuting the second
 and third diagonal entry of a diagonal matrix can also be obtained by
  proper rotations, and correspondingly for other permutations of the
 SSV. 
Hence all permutations of the SSV are allowed.\\ 

Fan \cite{Fan1951} gave bounds 
on the SV of $A+B$ given the SV of two real matrices $A$ and $B$,
derived from the corresponding results for eigenvalues of hermitian
matrices and using that the matrix $\tilde{A} \equiv \begin{pmatrix} 
0_{n\times n}& A\\A^T & 0_{n\times n}
\end{pmatrix} $ has the singular values of $A$ and their negatives as eigenvalues \cite[p.243 for review]{Marshall1979}. 
In the main part of this work we need a more constraining statement using the SSV, and thus taking the determinant of $A$, $B$, and $A+B$ into account as well. This leads to  theorem \ref{theo:SSV}. In the following we will denote with $\tilde{\boldsymbol{\sigma}}(A)$ the vector of canonical 
 SSV of the $n\times n$ real matrix A. Since the product of two rotations is again a rotation it follows directly from \eqref{def:App_SSV} that 
\begin{align}
\tilde{\boldsymbol{\sigma}}(Q_1A Q_2) = \tilde{\boldsymbol{\sigma}}(A), \; \forall Q_1,Q_2 \in SO(n). \label{eq:invariance}
\end{align}
Let $\boldsymbol{w}$ be a $n$-dimensional vector. We define 
\begin{align}
\Delta_{\boldsymbol{w}}\equiv \text{Conv}\left( \left\{\left.\left( s_1 w_{\pi(1)}, ..., s_n w_{\pi(n)}\right)^T\right|s_\nu \in \{-1,1\}: \prod_\nu s_\nu =1, \pi\in S_n\right\}\right) \label{def:Delta}
\end{align}
as the convex hull of all possible permutations $\pi\in S_n$ of the
components of $\boldsymbol{w}$ 
multiplied with an even number of minus signs. 
Let now $\boldsymbol{w_1}$ and $\boldsymbol{w_2}$ be two $n$-dimensional vectors. We define 
\begin{align}
\Sigma_{\boldsymbol{w_1}, \boldsymbol{w_2}} \equiv \left\{ a+b| a\in \Delta_{\boldsymbol{w_1}}, b\in \Delta_{\boldsymbol{w_2}}\right\}. \label{def:Sigma}
\end{align}
Figure \ref{Fig:sets} presents an illustration of the case
$n=2$.

\begin{figure}
\centering
\includegraphics[scale=1.5]{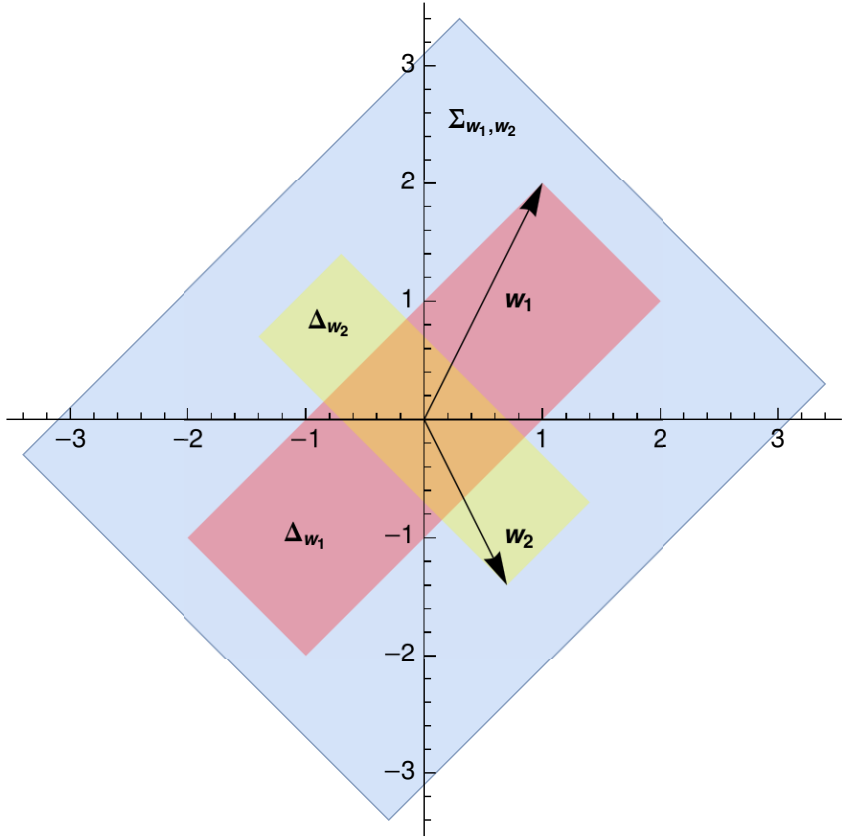}
\caption{\textbf{Illustration of theorem \ref{theo:SSV}:} Suppose we have two $2\times 2$ matrices $A$ and $B$ with SSV $\boldsymbol{w}_1$ and $\boldsymbol{w}_2$ respectively. The red and the yellow sets correspond to $\Sigma_{\boldsymbol{w}_1}$ and $\Sigma_{\boldsymbol{w}_2}$ defined by \eqref{def:Delta}. By theorem \ref{theo:SSV} the vector of SSV of $A+B$ then lies within the blue set, defined by \eqref{def:Sigma}.}
\label{Fig:sets}
\end{figure}
\begin{theo} \label{theo:SSV} Let $A$ and $B$ be two 
$n\times n$ real matrices whose SSV are known. Then
\begin{align}
\tilde{\boldsymbol{\sigma}}(A+B) \in \Sigma_{\tilde{\boldsymbol{\sigma}}(A),\tilde{\boldsymbol{\sigma}}(B)}.
\end{align}
\end{theo}
\begin{proof} Let $A$ be a 
$n\times n$ real matrix and let $\boldsymbol{d}(A)$ denote the vector of diagonal entries of $A$. Thompson showed  the following two statements about the diagonal elements of $A$ \cite[theorems 7 and 8]{Thompson1977} 
\begin{align}
 i)& \boldsymbol{d}(A) \in \Delta_{\tilde{\boldsymbol{\sigma}}(A)},\label{eq:diag1}\\
ii)& \forall \boldsymbol{d} \in \Delta_{\tilde{\boldsymbol{\sigma}}(A)}\; \exists R_1,R_2 \in SO(n): \boldsymbol{d}= \boldsymbol{d}(R_1AR_2). \label{eq:diag2}
\end{align} 
Now let $A$ and $B$ be two $n\times n$ real matrices.  Let $R_1,R_2 \in SO(n)$ such that $\boldsymbol{d}(R_1(A+B)R_2) = \tilde{\boldsymbol{\sigma}}(A+B)$. We then have
\begin{align*}
 \tilde{\boldsymbol{\sigma}}(A+B) = \boldsymbol{d}(R_1(A+B)R_2) =  \boldsymbol{d}(R_1AR_2) + \boldsymbol{d}(R_1BR_2)  \in   \Sigma_{\tilde{\boldsymbol{\sigma}}(R_1AR_2) ,\tilde{\boldsymbol{\sigma}}(R_1BR_2) } = \Sigma_{\tilde{\boldsymbol{\sigma}}(A) ,\tilde{\boldsymbol{\sigma}}(B) },
\end{align*}
where the second equation follows from the linearity of matrix addition in every element and the last equality from \eqref{eq:invariance}.
\end{proof}\\
As mentioned above, results for the absolute singular values of $A+B$
have been known before. To complete, we show that the above proof works analogously for the corresponding statement on absolute singular values: Let
$\boldsymbol{\sigma}(A)$ denote the vector of canonical absolute singular values
of an $n\times n$ real matrix $A$, $\sigma_1(A) \geq \sigma_2(A)
\geq ... \geq \sigma_n(A)$. Let $B$ be another $n\times n$ real
matrix. Then \cite[chapter 9 G.1.d.]{Marshall1979} 
\begin{align}
\boldsymbol{\sigma}(A+B) \prec_w \boldsymbol{\sigma}(A) + \boldsymbol{\sigma}(B), \label{eq:majorizaiton}
\end{align}
i.e.~the vector of canonical singular values of $A+B$  is \textit{weakly majorized} by the sum of the
 vectors of canonical singular values of $A$ and $B$. \textit{Weak majorization} for two vectors $\boldsymbol{x}$ and $\boldsymbol{y}$ with $x_1\geq x_2 \geq ...\geq x_n$ and $y_1 \geq y_2\geq...\geq y_n$ is defined as
\begin{align}
\boldsymbol{x} \prec_w \boldsymbol{y} \Leftrightarrow \sum_{i=1}^k x_i \leq \sum_{i=1}^k y_k\; \forall k \in \left\{1,2,...,n\right\}.
\end{align}
To see 
\eqref{eq:majorizaiton} define $\Delta'_{\boldsymbol{w}}$ analogously
to \eqref{def:Delta} but without the constraint $ \prod_\nu s_\nu =1$, i.e.~allowing arbitrary sign flips.  The
analogue statements of \eqref{eq:diag1} and \eqref{eq:diag2} hold if we
exchange the SSV with the absolute singular values, proper rotations (elements of $SO(n)$) with
orthogonal matrices (elements of $O(n)$),  and $\Delta_{\boldsymbol{w}}$ with
$\Delta'_{\boldsymbol{w}}$. We then find, that
$\boldsymbol{\sigma}(A+B) \in
\Sigma'_{\boldsymbol{\sigma}(A),\boldsymbol{\sigma}(B)}$, with $\Sigma'_{\boldsymbol{w_1}, \boldsymbol{w_2}} \equiv \left\{ a+b| a\in \Delta'_{\boldsymbol{w_1}}, b\in \Delta'_{\boldsymbol{w_2}}\right\}$. 
Since per definition the absolute singular values are non-negative, we
can further restrict $\Sigma'$ to the first hyperoctant. On the other
hand, for two  vectors $\boldsymbol{x}$, $\boldsymbol{y} \in
\mathbb{R}^n_+$ we have (proposition C.2. of chapter 4 in 
\cite{Marshall1979})
\begin{align}
\boldsymbol{x} \prec_w \boldsymbol{y} \Leftrightarrow \boldsymbol{x} \in \text{Conv}\left(\left\{s_1 y_{\pi(1)},..., s_n y_{\pi(n)}| s_\nu \in \{0,1\}, \pi\in S_n\right\}\right).
\end{align} 
The set on the r.h.s.~coincides with the restriction of $\Sigma'$ to the first hyperoctant if we take $\boldsymbol{y} = \boldsymbol{\sigma}(A) + \boldsymbol{\sigma}(B)$. Taking $\boldsymbol{x} = \boldsymbol{\sigma}(A+B)$, eq.~\eqref{eq:majorizaiton} follows.

\bibliography{library}

\begin{thebibliography}{26}%
\makeatletter
\providecommand \@ifxundefined [1]{%
 \@ifx{#1\undefined}
}%
\providecommand \@ifnum [1]{%
 \ifnum #1\expandafter \@firstoftwo
 \else \expandafter \@secondoftwo
 \fi
}%
\providecommand \@ifx [1]{%
 \ifx #1\expandafter \@firstoftwo
 \else \expandafter \@secondoftwo
 \fi
}%
\providecommand \natexlab [1]{#1}%
\providecommand \enquote  [1]{``#1''}%
\providecommand \bibnamefont  [1]{#1}%
\providecommand \bibfnamefont [1]{#1}%
\providecommand \citenamefont [1]{#1}%
\providecommand \href@noop [0]{\@secondoftwo}%
\providecommand \href [0]{\begingroup \@sanitize@url \@href}%
\providecommand \@href[1]{\@@startlink{#1}\@@href}%
\providecommand \@@href[1]{\endgroup#1\@@endlink}%
\providecommand \@sanitize@url [0]{\catcode `\\12\catcode `\$12\catcode
  `\&12\catcode `\#12\catcode `\^12\catcode `\_12\catcode `\%12\relax}%
\providecommand \@@startlink[1]{}%
\providecommand \@@endlink[0]{}%
\providecommand \url  [0]{\begingroup\@sanitize@url \@url }%
\providecommand \@url [1]{\endgroup\@href {#1}{\urlprefix }}%
\providecommand \urlprefix  [0]{URL }%
\providecommand \Eprint [0]{\href }%
\providecommand \doibase [0]{http://dx.doi.org/}%
\providecommand \selectlanguage [0]{\@gobble}%
\providecommand \bibinfo  [0]{\@secondoftwo}%
\providecommand \bibfield  [0]{\@secondoftwo}%
\providecommand \translation [1]{[#1]}%
\providecommand \BibitemOpen [0]{}%
\providecommand \bibitemStop [0]{}%
\providecommand \bibitemNoStop [0]{.\EOS\space}%
\providecommand \EOS [0]{\spacefactor3000\relax}%
\providecommand \BibitemShut  [1]{\csname bibitem#1\endcsname}%
\let\auto@bib@innerbib\@empty
\bibitem [{\citenamefont {Reichenbach}(1971)}]{Reichenbach1971}%
  \BibitemOpen
  \bibfield  {author} {\bibinfo {author} {\bibfnamefont {H.}~\bibnamefont
  {Reichenbach}},\ }\href@noop {} {\emph {\bibinfo {title} {{The direction of
  time}}}}\ (\bibinfo  {publisher} {University of California Press},\ \bibinfo
  {address} {Berkeley},\ \bibinfo {year} {1971})\BibitemShut {NoStop}%
\bibitem [{\citenamefont {Ried}\ \emph {et~al.}(2015)\citenamefont {Ried},
  \citenamefont {Agnew}, \citenamefont {Vermeyden}, \citenamefont {Janzing},
  \citenamefont {Spekkens},\ and\ \citenamefont {Resch}}]{Ried2015}%
  \BibitemOpen
  \bibfield  {author} {\bibinfo {author} {\bibfnamefont {K.}~\bibnamefont
  {Ried}}, \bibinfo {author} {\bibfnamefont {M.}~\bibnamefont {Agnew}},
  \bibinfo {author} {\bibfnamefont {L.}~\bibnamefont {Vermeyden}}, \bibinfo
  {author} {\bibfnamefont {D.}~\bibnamefont {Janzing}}, \bibinfo {author}
  {\bibfnamefont {R.~W.}\ \bibnamefont {Spekkens}}, \ and\ \bibinfo {author}
  {\bibfnamefont {K.~J.}\ \bibnamefont {Resch}},\ }\href {\doibase
  10.1038/nphys3266} {\bibfield  {journal} {\bibinfo  {journal} {Nat. Phys.}\
  }\textbf {\bibinfo {volume} {11}},\ \bibinfo {pages} {414} (\bibinfo {year}
  {2015})},\ \Eprint {http://arxiv.org/abs/1406.5036} {arXiv:1406.5036}
  \BibitemShut {NoStop}%
\bibitem [{\citenamefont {Bengtsson}\ and\ \citenamefont
  {{\.{Z}}yczkowski}(2006)}]{Bengtsson2006}%
  \BibitemOpen
  \bibfield  {author} {\bibinfo {author} {\bibfnamefont {I.}~\bibnamefont
  {Bengtsson}}\ and\ \bibinfo {author} {\bibfnamefont {K.}~\bibnamefont
  {{\.{Z}}yczkowski}},\ }\href {\doibase
  http://dx.doi.org/10.1017/CBO9780511535048} {\emph {\bibinfo {title}
  {{Geometry of quantum states: an introduction to quantum entanglement}}}}\
  (\bibinfo  {publisher} {Cambridge University Press},\ \bibinfo {address}
  {Cambridge},\ \bibinfo {year} {2006})\BibitemShut {NoStop}%
\bibitem [{\citenamefont {Nielsen}\ and\ \citenamefont
  {Chuang}(2010)}]{Nielsen2009}%
  \BibitemOpen
  \bibfield  {author} {\bibinfo {author} {\bibfnamefont {M.~A.}\ \bibnamefont
  {Nielsen}}\ and\ \bibinfo {author} {\bibfnamefont {I.~L.}\ \bibnamefont
  {Chuang}},\ }\href@noop {} {\emph {\bibinfo {title} {{Quantum Computation and
  Quantum Information}}}},\ \bibinfo {edition} {10th}\ ed.\ (\bibinfo
  {publisher} {Cambridge University Press},\ \bibinfo {address} {Cambridge},\
  \bibinfo {year} {2010})\BibitemShut {NoStop}%
\bibitem [{\citenamefont {Pearl}(2009)}]{Pearl2009a}%
  \BibitemOpen
  \bibfield  {author} {\bibinfo {author} {\bibfnamefont {J.}~\bibnamefont
  {Pearl}},\ }\href {\doibase 10.1215/00318108-110-4-639} {\emph {\bibinfo
  {title} {{Causality}}}},\ \bibinfo {edition} {2nd}\ ed.\ (\bibinfo
  {publisher} {Cambridge University Press},\ \bibinfo {address} {Cambridge},\
  \bibinfo {year} {2009})\BibitemShut {NoStop}%
\bibitem [{\citenamefont {Mooij}\ \emph {et~al.}(2016)\citenamefont {Mooij},
  \citenamefont {Peters}, \citenamefont {Janzing}, \citenamefont
  {Zscheischler},\ and\ \citenamefont {Sch{\"{o}}lkopf}}]{Mooij2016}%
  \BibitemOpen
  \bibfield  {author} {\bibinfo {author} {\bibfnamefont {J.~M.}\ \bibnamefont
  {Mooij}}, \bibinfo {author} {\bibfnamefont {J.}~\bibnamefont {Peters}},
  \bibinfo {author} {\bibfnamefont {D.}~\bibnamefont {Janzing}}, \bibinfo
  {author} {\bibfnamefont {J.}~\bibnamefont {Zscheischler}}, \ and\ \bibinfo
  {author} {\bibfnamefont {B.}~\bibnamefont {Sch{\"{o}}lkopf}},\ }\href
  {\doibase 10.1109/TSE.2014.2322358} {\bibfield  {journal} {\bibinfo
  {journal} {J. Mach. Learn. Res.}\ }\textbf {\bibinfo {volume} {17}},\
  \bibinfo {pages} {1} (\bibinfo {year} {2016})},\ \Eprint
  {http://arxiv.org/abs/1412.3773} {arXiv:1412.3773} \BibitemShut {NoStop}%
\bibitem [{\citenamefont {Chiribella}\ \emph {et~al.}(2011)\citenamefont
  {Chiribella}, \citenamefont {D'Ariano},\ and\ \citenamefont
  {Perinotti}}]{Chiribella2011}%
  \BibitemOpen
  \bibfield  {author} {\bibinfo {author} {\bibfnamefont {G.}~\bibnamefont
  {Chiribella}}, \bibinfo {author} {\bibfnamefont {G.~M.}\ \bibnamefont
  {D'Ariano}}, \ and\ \bibinfo {author} {\bibfnamefont {P.}~\bibnamefont
  {Perinotti}},\ }\href {\doibase 10.1103/PhysRevA.84.012311} {\bibfield
  {journal} {\bibinfo  {journal} {Phys. Rev. A}\ }\textbf {\bibinfo {volume}
  {84}},\ \bibinfo {pages} {012311} (\bibinfo {year} {2011})},\ \Eprint
  {http://arxiv.org/abs/1011.6451} {arXiv:1011.6451} \BibitemShut {NoStop}%
\bibitem [{\citenamefont {Oreshkov}\ \emph {et~al.}(2012)\citenamefont
  {Oreshkov}, \citenamefont {Costa},\ and\ \citenamefont {Brukner}}]{Oreshkov}%
  \BibitemOpen
  \bibfield  {author} {\bibinfo {author} {\bibfnamefont {O.}~\bibnamefont
  {Oreshkov}}, \bibinfo {author} {\bibfnamefont {F.}~\bibnamefont {Costa}}, \
  and\ \bibinfo {author} {\bibfnamefont {{\v{C}}.}~\bibnamefont {Brukner}},\
  }\href {\doibase 10.1038/ncomms2076} {\bibfield  {journal} {\bibinfo
  {journal} {Nat. Commun.}\ }\textbf {\bibinfo {volume} {3}},\ \bibinfo {pages}
  {1092} (\bibinfo {year} {2012})},\ \Eprint {http://arxiv.org/abs/1105.4464v3}
  {arXiv:1105.4464v3} \BibitemShut {NoStop}%
\bibitem [{\citenamefont {Costa}\ and\ \citenamefont
  {Shrapnel}(2016)}]{Costa2016}%
  \BibitemOpen
  \bibfield  {author} {\bibinfo {author} {\bibfnamefont {F.}~\bibnamefont
  {Costa}}\ and\ \bibinfo {author} {\bibfnamefont {S.}~\bibnamefont
  {Shrapnel}},\ }\href {\doibase 10.1088/1367-2630/18/6/063032} {\bibfield
  {journal} {\bibinfo  {journal} {New J. Phys.}\ }\textbf {\bibinfo {volume}
  {18}},\ \bibinfo {pages} {063032} (\bibinfo {year} {2016})},\ \Eprint
  {http://arxiv.org/abs/1512.07106} {arXiv:1512.07106} \BibitemShut {NoStop}%
\bibitem [{\citenamefont {Oreshkov}\ and\ \citenamefont
  {Giarmatzi}(2016)}]{Oreshkov2016}%
  \BibitemOpen
  \bibfield  {author} {\bibinfo {author} {\bibfnamefont {O.}~\bibnamefont
  {Oreshkov}}\ and\ \bibinfo {author} {\bibfnamefont {C.}~\bibnamefont
  {Giarmatzi}},\ }\href {\doibase 10.1088/1367-2630/18/9/093020} {\bibfield
  {journal} {\bibinfo  {journal} {New J. Phys.}\ }\textbf {\bibinfo {volume}
  {18}},\ \bibinfo {pages} {093020} (\bibinfo {year} {2016})},\ \Eprint
  {http://arxiv.org/abs/1506.05449} {arXiv:1506.05449} \BibitemShut {NoStop}%
\bibitem [{\citenamefont {Procopio}\ \emph {et~al.}(2015)\citenamefont
  {Procopio}, \citenamefont {Moqanaki}, \citenamefont {Araujo}, \citenamefont
  {Costa}, \citenamefont {{Alonso Calafell}}, \citenamefont {Dowd},
  \citenamefont {Hamel}, \citenamefont {Rozema}, \citenamefont {Brukner},\ and\
  \citenamefont {Walther}}]{Procopio2015}%
  \BibitemOpen
  \bibfield  {author} {\bibinfo {author} {\bibfnamefont {L.~M.}\ \bibnamefont
  {Procopio}}, \bibinfo {author} {\bibfnamefont {A.}~\bibnamefont {Moqanaki}},
  \bibinfo {author} {\bibfnamefont {M.}~\bibnamefont {Araujo}}, \bibinfo
  {author} {\bibfnamefont {F.}~\bibnamefont {Costa}}, \bibinfo {author}
  {\bibfnamefont {I.}~\bibnamefont {{Alonso Calafell}}}, \bibinfo {author}
  {\bibfnamefont {E.~G.}\ \bibnamefont {Dowd}}, \bibinfo {author}
  {\bibfnamefont {D.~R.}\ \bibnamefont {Hamel}}, \bibinfo {author}
  {\bibfnamefont {L.~A.}\ \bibnamefont {Rozema}}, \bibinfo {author}
  {\bibfnamefont {C.}~\bibnamefont {Brukner}}, \ and\ \bibinfo {author}
  {\bibfnamefont {P.}~\bibnamefont {Walther}},\ }\href {\doibase
  10.1038/ncomms8913} {\bibfield  {journal} {\bibinfo  {journal} {Nat.
  Commun.}\ }\textbf {\bibinfo {volume} {6}},\ \bibinfo {pages} {7913}
  (\bibinfo {year} {2015})},\ \Eprint {http://arxiv.org/abs/1412.4006}
  {arXiv:1412.4006} \BibitemShut {NoStop}%
\bibitem [{\citenamefont {Chiribella}(2012)}]{Chiribella2012}%
  \BibitemOpen
  \bibfield  {author} {\bibinfo {author} {\bibfnamefont {G.}~\bibnamefont
  {Chiribella}},\ }\href {\doibase 10.1103/PhysRevA.86.040301} {\bibfield
  {journal} {\bibinfo  {journal} {Phys. Rev. A}\ }\textbf {\bibinfo {volume}
  {86}},\ \bibinfo {pages} {040301} (\bibinfo {year} {2012})},\ \Eprint
  {http://arxiv.org/abs/1109.5154v3} {arXiv:1109.5154v3} \BibitemShut {NoStop}%
\bibitem [{\citenamefont {Gu{\'{e}}rin}\ \emph {et~al.}(2016)\citenamefont
  {Gu{\'{e}}rin}, \citenamefont {Feix}, \citenamefont {Ara{\'{u}}jo},\ and\
  \citenamefont {Brukner}}]{Guerin}%
  \BibitemOpen
  \bibfield  {author} {\bibinfo {author} {\bibfnamefont {P.~A.}\ \bibnamefont
  {Gu{\'{e}}rin}}, \bibinfo {author} {\bibfnamefont {A.}~\bibnamefont {Feix}},
  \bibinfo {author} {\bibfnamefont {M.}~\bibnamefont {Ara{\'{u}}jo}}, \ and\
  \bibinfo {author} {\bibfnamefont {{\v{C}}.}~\bibnamefont {Brukner}},\ }\href
  {\doibase 10.1103/PhysRevLett.117.100502} {\bibfield  {journal} {\bibinfo
  {journal} {Phys. Rev. Lett.}\ }\textbf {\bibinfo {volume} {117}},\ \bibinfo
  {pages} {100502} (\bibinfo {year} {2016})},\ \Eprint
  {http://arxiv.org/abs/1605.07372} {arXiv:1605.07372} \BibitemShut {NoStop}%
\bibitem [{\citenamefont {Fujiwara}\ and\ \citenamefont
  {Algoet}(1999)}]{Fujiwara1999}%
  \BibitemOpen
  \bibfield  {author} {\bibinfo {author} {\bibfnamefont {A.}~\bibnamefont
  {Fujiwara}}\ and\ \bibinfo {author} {\bibfnamefont {P.}~\bibnamefont
  {Algoet}},\ }\href {\doibase 10.1103/PhysRevA.59.3290} {\bibfield  {journal}
  {\bibinfo  {journal} {Phys. Rev. A}\ }\textbf {\bibinfo {volume} {59}},\
  \bibinfo {pages} {3290} (\bibinfo {year} {1999})}\BibitemShut {NoStop}%
\bibitem [{\citenamefont {Braun}\ \emph {et~al.}(2014)\citenamefont {Braun},
  \citenamefont {Giraud}, \citenamefont {Nechita}, \citenamefont {Pellegrini},\
  and\ \citenamefont {Znidaric}}]{Braun2014}%
  \BibitemOpen
  \bibfield  {author} {\bibinfo {author} {\bibfnamefont {D.}~\bibnamefont
  {Braun}}, \bibinfo {author} {\bibfnamefont {O.}~\bibnamefont {Giraud}},
  \bibinfo {author} {\bibfnamefont {I.}~\bibnamefont {Nechita}}, \bibinfo
  {author} {\bibfnamefont {C.~E.}\ \bibnamefont {Pellegrini}}, \ and\ \bibinfo
  {author} {\bibfnamefont {M.}~\bibnamefont {Znidaric}},\ }\href {\doibase
  10.1088/1751-8113/47/13/135302} {\bibfield  {journal} {\bibinfo  {journal}
  {J. Phys. A Math. Theor.}\ }\textbf {\bibinfo {volume} {47}},\ \bibinfo
  {pages} {135302} (\bibinfo {year} {2014})},\ \Eprint
  {http://arxiv.org/abs/1306.0495v2} {arXiv:1306.0495v2} \BibitemShut {NoStop}%
\bibitem [{\citenamefont {{Beth Ruskai}}\ \emph {et~al.}(2002)\citenamefont
  {{Beth Ruskai}}, \citenamefont {Szarek},\ and\ \citenamefont
  {Werner}}]{BethRuskai2002}%
  \BibitemOpen
  \bibfield  {author} {\bibinfo {author} {\bibfnamefont {M.}~\bibnamefont
  {{Beth Ruskai}}}, \bibinfo {author} {\bibfnamefont {S.}~\bibnamefont
  {Szarek}}, \ and\ \bibinfo {author} {\bibfnamefont {E.}~\bibnamefont
  {Werner}},\ }\href {\doibase 10.1016/S0024-3795(01)00547-X} {\bibfield
  {journal} {\bibinfo  {journal} {Linear Algebra Appl.}\ }\textbf {\bibinfo
  {volume} {347}},\ \bibinfo {pages} {159} (\bibinfo {year} {2002})},\ \Eprint
  {http://arxiv.org/abs/0101003v2} {arXiv:0101003v2 [quant-ph]} \BibitemShut
  {NoStop}%
\bibitem [{\citenamefont {Bell}(1964)}]{Bell1964}%
  \BibitemOpen
  \bibfield  {author} {\bibinfo {author} {\bibfnamefont {J.~S.}\ \bibnamefont
  {Bell}},\ }\href@noop {} {\bibfield  {journal} {\bibinfo  {journal}
  {Physics}\ }\textbf {\bibinfo {volume} {1}},\ \bibinfo {pages} {195}
  (\bibinfo {year} {1964})}\BibitemShut {NoStop}%
\bibitem [{\citenamefont {Schr{\"{o}}dinger}(1935)}]{Schrodinger1935}%
  \BibitemOpen
  \bibfield  {author} {\bibinfo {author} {\bibfnamefont {E.}~\bibnamefont
  {Schr{\"{o}}dinger}},\ }\href {\doibase 10.1017/S0305004100013554} {\bibfield
   {journal} {\bibinfo  {journal} {Math. Proc. Cambridge Phil. Soc.}\ }\textbf
  {\bibinfo {volume} {31}},\ \bibinfo {pages} {555} (\bibinfo {year}
  {1935})}\BibitemShut {NoStop}%
\bibitem [{\citenamefont {Jevtic}\ \emph {et~al.}(2014)\citenamefont {Jevtic},
  \citenamefont {Pusey}, \citenamefont {Jennings},\ and\ \citenamefont
  {Rudolph}}]{Jevtic}%
  \BibitemOpen
  \bibfield  {author} {\bibinfo {author} {\bibfnamefont {S.}~\bibnamefont
  {Jevtic}}, \bibinfo {author} {\bibfnamefont {M.}~\bibnamefont {Pusey}},
  \bibinfo {author} {\bibfnamefont {D.}~\bibnamefont {Jennings}}, \ and\
  \bibinfo {author} {\bibfnamefont {T.}~\bibnamefont {Rudolph}},\ }\href
  {\doibase 10.1103/PhysRevLett.113.020402} {\bibfield  {journal} {\bibinfo
  {journal} {Phys. Rev. Lett.}\ }\textbf {\bibinfo {volume} {113}},\ \bibinfo
  {pages} {020402} (\bibinfo {year} {2014})},\ \Eprint
  {http://arxiv.org/abs/1303.4724} {arXiv:1303.4724} \BibitemShut {NoStop}%
\bibitem [{\citenamefont {Milne}\ \emph {et~al.}(2014)\citenamefont {Milne},
  \citenamefont {Jevtic}, \citenamefont {Jennings}, \citenamefont {Wiseman},\
  and\ \citenamefont {Rudolph}}]{Milne2014}%
  \BibitemOpen
  \bibfield  {author} {\bibinfo {author} {\bibfnamefont {A.}~\bibnamefont
  {Milne}}, \bibinfo {author} {\bibfnamefont {S.}~\bibnamefont {Jevtic}},
  \bibinfo {author} {\bibfnamefont {D.}~\bibnamefont {Jennings}}, \bibinfo
  {author} {\bibfnamefont {H.}~\bibnamefont {Wiseman}}, \ and\ \bibinfo
  {author} {\bibfnamefont {T.}~\bibnamefont {Rudolph}},\ }\href {\doibase
  10.1088/1367-2630/16/8/083017} {\bibfield  {journal} {\bibinfo  {journal}
  {New J. Phys.}\ }\textbf {\bibinfo {volume} {16}},\ \bibinfo {pages} {083017}
  (\bibinfo {year} {2014})},\ \Eprint {http://arxiv.org/abs/1403.0418}
  {arXiv:1403.0418} \BibitemShut {NoStop}%
\bibitem [{\citenamefont {Ruskai}(2003)}]{Ruskai2003}%
  \BibitemOpen
  \bibfield  {author} {\bibinfo {author} {\bibfnamefont {M.~B.}\ \bibnamefont
  {Ruskai}},\ }\href {\doibase 10.1142/S0129055X03001710} {\bibfield  {journal}
  {\bibinfo  {journal} {Rev. Math. Phys.}\ }\textbf {\bibinfo {volume} {15}},\
  \bibinfo {pages} {643} (\bibinfo {year} {2003})},\ \Eprint
  {http://arxiv.org/abs/0302032} {arXiv:0302032 [quant-ph]} \BibitemShut
  {NoStop}%
\bibitem [{\citenamefont {Werner}(1989)}]{Werner1989}%
  \BibitemOpen
  \bibfield  {author} {\bibinfo {author} {\bibfnamefont {R.~F.}\ \bibnamefont
  {Werner}},\ }\href {\doibase 10.1103/PhysRevA.40.4277} {\bibfield  {journal}
  {\bibinfo  {journal} {Phys. Rev. A}\ }\textbf {\bibinfo {volume} {40}},\
  \bibinfo {pages} {4277} (\bibinfo {year} {1989})}\BibitemShut {NoStop}%
\bibitem [{\citenamefont {Amir-Moez}\ and\ \citenamefont
  {Horn}(1958)}]{Amir-Moez1958}%
  \BibitemOpen
  \bibfield  {author} {\bibinfo {author} {\bibfnamefont {A.~R.}\ \bibnamefont
  {Amir-Moez}}\ and\ \bibinfo {author} {\bibfnamefont {A.}~\bibnamefont
  {Horn}},\ }\href {http://www.jstor.org/stable/2310676
  http://about.jstor.org/terms} {\bibfield  {journal} {\bibinfo  {journal} {Am.
  Math. Mon.}\ }\textbf {\bibinfo {volume} {65}},\ \bibinfo {pages} {742}
  (\bibinfo {year} {1958})}\BibitemShut {NoStop}%
\bibitem [{\citenamefont {Fan}(1951)}]{Fan1951}%
  \BibitemOpen
  \bibfield  {author} {\bibinfo {author} {\bibfnamefont {K.}~\bibnamefont
  {Fan}},\ }\href {\doibase 10.1073/pnas.37.11.760} {\bibfield  {journal}
  {\bibinfo  {journal} {Proceedings of the National Academy of Sciences of the
  United States of America}\ }\textbf {\bibinfo {volume} {37}},\ \bibinfo
  {pages} {760} (\bibinfo {year} {1951})}\BibitemShut {NoStop}%
\bibitem [{\citenamefont {Marshall}\ and\ \citenamefont
  {Olkin}(1979)}]{Marshall1979}%
  \BibitemOpen
  \bibfield  {author} {\bibinfo {author} {\bibfnamefont {A.~W.}\ \bibnamefont
  {Marshall}}\ and\ \bibinfo {author} {\bibfnamefont {I.}~\bibnamefont
  {Olkin}},\ }\href {\doibase 10.1007/978-0-387-68276-1} {\emph {\bibinfo
  {title} {{Inequalities: Theory of Majorization and Its Applications}}}}\
  (\bibinfo  {publisher} {Academic Press},\ \bibinfo {address} {New York},\
  \bibinfo {year} {1979})\BibitemShut {NoStop}%
\bibitem [{\citenamefont {Thompson}(1977)}]{Thompson1977}%
  \BibitemOpen
  \bibfield  {author} {\bibinfo {author} {\bibfnamefont {R.}~\bibnamefont
  {Thompson}},\ }\href@noop {} {\bibfield  {journal} {\bibinfo  {journal} {SIAM
  J. Appl. Math.}\ }\textbf {\bibinfo {volume} {32}},\ \bibinfo {pages} {39}
  (\bibinfo {year} {1977})}\BibitemShut {NoStop}%
\end{thebibliography}%
\end{document}